\pdfoutput=1

\documentclass[11pt,letterpaper]{article}
\usepackage[hscale=0.7,vscale=0.8]{geometry}

\usepackage{mystyle} 
\title{Mechanism Design for Cumulative Prospect Theoretic Agents: A General Framework and the Revelation Principle}
\author{Soham R.\ Phade%
\thanks{Corresponding author}
\ and Venkat Anantharam
\thanks{The authors are with the Department of Electrical Engineering and Computer Sciences, University of California, Berkeley, Berkeley, CA 94720.
       {\tt\small soham\_phade@berkeley.edu, ananth@eecs.berkeley.edu}}%
\thanks{Research supported by the NSF Science and Technology Center grant CCF-
0939370: "Science of Information", the NSF grants ECCS-1343398, CNS-1527846, CIF-1618145,
CCF-1901004 
and CIF-2007965, 
and the William and Flora Hewlett Foundation supported Center for Long Term Cybersecurity at Berkeley.}
}
\date{}

\newif\ifvenkat
\venkatfalse

\newif\ifsoham
\sohamfalse

\newif\ifnonarxiv
\nonarxivfalse

\let\footnote=\endnote

\newcommand{\red}{\color{red}}
\newcommand{\black}{\color{black}}
\newcommand{\blue}{\color{blue}}
\newcommand{\brown}{\color{brown}}

\newcommand{\type}{\theta}
\newcommand{\Type}{\Theta}
\newcommand{\sig}{\psi}
\newcommand{\Sig}{\Psi}
\newcommand{\msg}{\phi}
\newcommand{\Msg}{\Phi}

\newcommand{\mdt}{D}

\newcommand{\cM}{\cal{M}}
\newcommand{\cF}{\Pi}

\newcommand{\I}{\mathrm{I}}
\newcommand{\II}{\mathrm{II}}
\newcommand{\III}{\mathrm{III}}
\newcommand{\IV}{\mathrm{IV}}
\newcommand{\V}{\mathrm{V}}

\newcommand{\Up}{\mathrm{UP}}
\newcommand{\D}{\mathrm{DN}}
\newcommand{\Lt}{\mathrm{LT}}
\newcommand{\Rt}{\mathrm{RT}}

\newcommand{\MF}{\mathrm{MF}}
\newcommand{\UF}{\mathrm{UF}}
\newcommand{\SF}{\mathrm{SF}}
\newcommand{\K}{\mathrm{K}}
\newcommand{\M}{\mathrm{M}}

\begin{document}
\maketitle


\begin{abstract}
	This paper initiates a discussion of mechanism design when the participating agents exhibit preferences that deviate from expected utility theory (EUT).
	In particular, we consider mechanism design for systems where 
	the 
	agents are modeled 
	as having 
	cumulative prospect theory (CPT) preferences, which is a generalization of EUT preferences.
	We point out some of the key 
	modifications needed in the theory of mechanism design that arise from 
	agents having 
	CPT preferences and some of the shortcomings of the classical mechanism design framework.
	In particular, we show that the revelation principle,
	which 
	has traditionally played a fundamental role in mechanism design, does not continue to hold under CPT.
	We develop an appropriate framework that we call mediated mechanism design which allows us to recover the revelation principle for CPT agents.
	We conclude with some interesting directions for future work.
\end{abstract}



\section{Introduction}

In nearly every application of mechanism design,
the decision-making entities are predominantly human beings faced with uncertainties.
These uncertainties, for example, could arise from a combination of one or more factors from the following:
\begin{inparaenum}[(i)]
	\item lack of information about the outcomes (e.g. oil lease auctions, kidney-exchange, insurance markets),
	\item each player 
	having 
	uncertainty about other players' behavior (e.g. voting behavior in elections, inclination to getting vaccinated in immunization programs),
	\item strategic interactions between the players (e.g. players could employ randomized strategies to hedge their market returns),
	\item randomness introduced by design (e.g. Tullock contests,
	where the probability of winning a prize depends on the amount of effort an agent puts into it%
	).
\end{inparaenum}
Naturally, to realize the mechanism designer's objectives, it is beneficial to consider as accurate and general models for human preference behavior under uncertainty as possible.
Cumulative prospect theory (CPT), proposed by \citet{tversky1992advances}, is one of the leading theories for 
decision making 
under uncertainty. 
Our goal here is to study mechanism design when players exhibit CPT preferences.

We are interested in situations where the \emph{agents} participating in the system have private \emph{types} 
(comprised of 
private information and preferences).
{\em Player} and agent are used as interchangeable terms. 
The \emph{system operator}%
\footnote{
\citet{myerson1982optimal} refers to the mechanism design framework as a generalized principal-agent problem. In contrast to our framework, Myerson is interested in problems where 
the 
agents have private decision domains in addition to private information. Here, by private decision domains, we mean possible actions for the player that directly affect the outcomes. For example, in employment contracts the actions of the employee directly 
affect 
the outcome. 
These actions should not be confused with the signals of the player in the communication protocol set up by the system operator. 
We prefer to call the entity in control as the system operator instead of the principal to emphasize that the system operator alone controls the system 
implementation. 
We restrict ourselves to situations where agents do not have private decision domains because such situations involve dynamic decision-making, and non-EUT models face several issues in such situations (see section~\ref{sec: discussion} for more on this).
Thus our model cannot account for moral hazard.
} 
is in a position to set the rules of 
communication 
and can control the 
implementation 
in the system.
It aims to achieve certain goals, such as social welfare or revenue generation, without getting to directly observe the types of the players.
Studying these systems when agents have CPT preferences requires modifications to the formal structures commonly encountered in classical mechanism design \citep{harsanyi1967games,myerson1979incentive,myerson1982optimal, myerson2004comments, mas1995microeconomic}.
But before engaging in a systematic discussion of these issues,
let us briefly describe our key result.

This starts with the observation that if the players are assumed to have CPT preferences instead of expected utility theory (EUT) preferences, then the revelation principle~\citep{myerson1981optimal}, one of the fundamental principles in mechanism design, does not hold anymore.
A related observation was made by \citet{karni1989dynamic}, where the authors show that in a second-price sealed-bid auction the revelation principle holds in general if and only if the players have EUT preferences.
\citet{chew1985implicit} provides an example to show that the revelation principle fails in a second-price sealed-bid auction when the players have preferences given by implicit weighted utility theory~\citep{dekel1986axiomatic,chew1989axiomatic}.



The classical mechanism design framework 
is comprised of a fixed number of players,
an allocation set, a set of types for each player, and a signal set for each
player.
(In this paper, we will be concerned with the setting where all these sets are assumed to be finite.)
The system operator commits to an allocation function, i.e. a function from 
the signal profile of the players to an allocation (see \eqref{def: allocation_function} for the formal definition). 

The mechanism operates as follows:
\begin{enumerate}
	\item Each player sends a signal
	strategically to the system operator based on 
 	its type
 	(which is private knowledge to the player).
	\item The system operator implements the allocation based on the signals from all the players in accordance to the allocation function that it committed to.
\end{enumerate}

If we assume a 
prior over the types of the players which is common knowledge to all the agents and the system operator,
and we assume that the signal sets of all the players, the allocation set and the allocation mapping are also common knowledge,
then this constitutes a Bayesian game and one studies the outcome of such a game through its Bayes-Nash equilibria (see~\eqref{eq: def_BNE} for the formal definition).
The revelation principle states that for the question of  implementability of social choice functions (see \eqref{eq: def_scf} and \eqref{eq: def_implementable} for formal definitions of social choice functions and their implementability), it is enough to assume the signal set to be the same as the type set for each player and confine attention to the equilibrium in which each player reports her type truthfully.
We propose a modification to the above framework that we call a \emph{mediated mechanism}.
We introduce a new stage where the system operator acts like a mediator and sends each player a private message sampled from a certain joint distribution on the set of message profiles. The allocation chosen by the system operator can now depend on both the message profile and the signal profile. 
Further, we explicitly allow the choice of the allocation to be randomized, which turns out to have no advantage in the classical mechanism design framework but can lead to benefits with CPT agents.

A mediated mechanism is therefore comprised of 
 a fixed number of players,
an allocation set, a set of types for each player, 
a message set for each player, and a signal set for each
player, all of which 
are generally assumed to be finite sets.
The system operator 
commits to 
a mediator distribution, which is a probability distribution on the set of 
message profiles. 
It also commits to a mediated allocation function, which maps each pair of signal profile and message profile to a probability distribution on allocations (see~\eqref{eq: def_med_allocation_func} for the formal definition).

The mechanism operates as follows:
\begin{enumerate}
	\item The system operator samples a message profile from the declared mediator distribution and sends the individual messages to each player privately.
	\item Each player receives her mediator message and, based on this message and her privately known type, sends a signal strategically to the 
	system operator. 
	\item Based on the signals collected from all the players and the sampled message profile, the system operator samples the allocation in accordance to the probability distribution on allocations resulting from the mediated allocation function that it committed to.
\end{enumerate}

Similarly to the previous setting we assume a prior over the types of the players that is common knowledge to all the agents and the system operator.
We also assume that the message sets and the signal sets of all the players, the mediator distribution, the allocation set, and the allocation-outcome mapping are common knowledge.
This along with the mechanism operation stated above constitutes a Bayesian game and we study the outcome of such a game through its Bayes-Nash equilibria (see~\eqref{eq: def_med_Bayes_Nash_eq} for the formal definition).
With this modified framework, 
we recover a form of the revelation principle which states that 
it is enough to assume the signal set to be the same as the type set 
for each player and confine our attention to the equilibrium in which each player reports her true type 
irrespective of the private message she receives from the mediator.
(See statement (i) of Theorem~\ref{thm: rev_prin_1}.)

As the mediator message 
sets
could be arbitrary, 
it might seem that the problem of designing the signal sets has been transformed into the problem of designing the message sets.
Although this is true, 
notice that the revelation principle allows us to restrict our attention to truthful strategies for each player, which have a simple form, 
thus resolving the difficult task of finding all the Bayes-Nash equilibria 
of the resulting game. 
Further, the fact that truthful reporting does not depend on the private message received by a player makes it a practical and natural strategy for the players.

We now resume our discussion of the different aspects involved in the study of mechanism design when agents have CPT preferences.
The majority of 
the
mechanism design literature
has been restricted to EUT
modeling of individual decision-making under uncertainty.
Indeed, EUT has a nice normative interpretation and provides a useful and insightful first-order approximation (see, for example, \citet{schoemaker1982expected}). 
However, systematic deviations from the predictions of EUT have been observed in several empirical studies involving human decision-makers~\citep{allais1953comportement,fishburn1979two,kahneman2013prospect} (see~\citet{starmer2000developments} for an excellent survey).
With the advent of e-commerce activities and the ever-growing online marketplaces such as Amazon, eBay, and Uber, 
where the participating agents are largely human beings, who exhibit behavior that is highly susceptible to these deviations from EUT, 
it has become crucial to account for such behavioral deviations in the modeling of these systems. (For example, \citet{pavlou2006nature} discuss the phenomenon of premium prices showing up in online marketplaces such as eBay to differentiate among sellers based on their reputation and buyers' perceived risks.)

A typical environment in the traditional mechanism design setup 
consists of 
a set of players that have private information about their types and an allocation set 
listing 
the possible alternatives from which the system operator chooses one that is best suited given the players' types.
As mentioned earlier, we assume that the system operator controls the 
implementation 
and 
the players do not have separate decision domains.
(Recall that by private decision domains we mean possible actions for the player that directly affect the outcomes.) 
This is typical in several online marketplaces.
For example, in online advertising platforms such as Google Ads, the platform has complete control over where to place which ads.
Note that although the agents can affect the implementation of the system through their bids, these signals fall under the communication protocol set by the system, leaving the ultimate implementation in the hands of the system operator. 
In online matching markets such as eBay and Uber, the platform matches the buyers to sellers as in eBay, or riders to drivers as in Uber.%
\footnote{
Here, strictly speaking, given a matching by the platform, the users can refuse to go through with the matching.
Although these decisions fall under the separate decision domains of the agents, they are rare and can be accounted for separately.
}

Even if the system operator has complete control over 
the implementation, 
it wants 
the implementation 
to depend on the types of the agents.
However, it does not have access to 
these types, 
and hence needs to design a mechanism to achieve this goal.
Thus the system implementation indirectly depends on the choices of the participating agents.
Note that the e-commerce applications mentioned above---Amazon, eBay, and Uber---fit well in this setup.
Indeed, these are instances of a delivery system, an auction house, and a clearinghouse, which have been 
topics of 
interest for several years in mechanism design.
However, the nature of these applications, and the presence of vast data corresponding to several repeated short-lived interactions of the system with any given user, makes it feasible to incorporate the behavioral features displayed by the users.

It has been a convention to assume that the outcome set for each player is identical to the allocation set, 
and hence the type for each player is assumed to capture her preferences over the allocation set (see, for example, \citet{vohra2011mechanism}).
However, in principle, the outcome set for any player need not be the same as the allocation set.
Indeed the allocation set is a list of the alternatives available to the system operator to implement, whereas the outcome set consists of the outcomes realized by the players, and these can be quite different.
For example, in the case of Amazon, the allocation set consists of alternative resource allocations to fulfill the delivery of purchased products, whereas the outcome set of a buyer consists of features such as time of delivery, place of delivery, etc.
It makes sense to consider the preferences of a player over her outcome set, and any consideration of her preferences over the allocation set should be thought of as a pullback or a precomposition of her preferences over the outcome set with respect to the
(possibly random) 
function that maps allocations to outcomes for this player.

We allow the above mapping from allocations to outcomes for any player to be randomized.
Indeed, more often than not, the system operator does not have complete control over the outcomes of the players due to intrinsic uncertainties present in the system.
For example, fixed resource allocations by Amazon can lead to uncertainty in the delivery times, possibly due to factors not part of the system model.
In the case of Uber, upon matching the riders with the drivers in a certain way and choosing their corresponding routes, the arrival times and the riding experience of the users remain uncertain.
In an auction setting such as eBay, if we consider the outcome set for any player to indicate if she receives the item or not, then the mapping from allocation to outcomes is deterministic.
However, if we model the outcome set to indicate whether the player is satisfied with the item she receives, then we have to allow the mapping to be randomized.

Furthermore, the system operator might not be able to observe the outcome realization, for example the ride experience of a passenger.
It can only try to learn this in hindsight through customer feedback.
Besides, the outcome set for any single player is typically small as compared to the allocation set and the product of the outcome spaces of all the participating agents.
Thus, treating each player's outcome set separately would enable us to focus on the preference behavior of an individual player and have better models for this player's preferences.

The (random) mapping from allocations to outcomes for any player induces a lottery $L$ on the outcome set of this player for each allocation.
EUT satisfies the linearity property which states that $U(\alpha L_1 + (1-\alpha)L_2) = \alpha U(L_1) + (1-\alpha) U(L_2)$, where $0 \leq \alpha \leq 1$, $L_1, L_2$ are two lotteries, and $U(\cdot)$ denotes the expected utility of the lottery within the parentheses.
This property of EUT allows us to model the type of a player by considering 
her 
utility values for each allocation. 
For any lottery $L$ over her outcomes that is induced by a lottery over the allocations $\mu$, we can evaluate her utility $U(L)$ by taking the expectation over her utility values of the allocations with respect to the distribution $\mu$.
CPT on the other hand does not satisfy this linearity property (see, for example, \citet{tversky1992advances}), and hence
it is important that we consider the general model with separate outcome sets.

We formalize this general setup in subsection~\ref{subsec: setup} and provide preliminary background on CPT preferences in subsection~\ref{subsec: CPTpref}.
Then, in subsection~\ref{subsec: mechframe}, we consider the traditional mechanism design framework where each player knows her (private) type and strategically sends a signal to the system operator.
The system operator collects these signals and implements a lottery over the allocation set.

We define a social choice function as 
a
function mapping each type profile into a lottery over the product of the outcome sets for each player
(see~\eqref{eq: def_scf} for the formal definition).
As an intermediate step, we consider an allocation choice function (i.e. a function that maps type profiles into lotteries over the allocation set, see~\eqref{eq: def_acf}).
Each allocation choice function uniquely defines a social choice function through the allocation-outcome 
mapping 
(see~\eqref{eq: acf_to_scf}),
which we think of as a mapping from allocations to probability distributions on the product of the outcome sets
of the agents. 
Note that there can be multiple allocation choice functions that give rise to the same social choice function.
We define the notion of implementability for an allocation choice function in Bayes-Nash equilibrium (see~\eqref{eq: def_implementable}). 
We say that a social choice function is implementable in Bayes-Nash equilibrium if there exists an allocation choice function that is implementable in Bayes-Nash equilibrium and induces this social choice function.

We similarly define the notions of implementability in dominant equilibrium.
Here, we identify an additional notion of implementability that we call implementable in belief-dominant equilibrium.
Roughly speaking, a dominant strategy is a best response to all the strategy profiles of the opponents (see~\eqref{eq: dominant_strat}), and a belief-dominant strategy is a best response to all the beliefs over the strategy profiles of the opponents (see~\eqref{eq: belief_dominant_strat}).
Under EUT, 
the notion of a 
dominant strategy is equivalent to 
that of a 
belief-dominant strategy.
However, this is not true in general
when the agents have CPT preferences, 
thus making it necessary to distinguish between these two notions of equilibrium.

In section~\ref{sec: rev_prin}, we define the notions of direct mechanism (see~\eqref{eq: def_direct_allocation_func}) and truthful implementation (see~\eqref{eq: def_dir_strat}).
We then give an example that highlights the shortcoming of restricting oneself to direct mechanisms when the players have CPT preferences, as opposed to EUT preferences.
In particular, we consider a 
$2$-player 
setting where the players have CPT preferences that are not EUT preferences.
Example~\ref{ex: no_rev_impl} gives an allocation choice function for which the 
revelation principle does not hold for implementation in Bayes-Nash equilibrium.
We then introduce the framework of mediated mechanism design in subsection~\ref{subsec: med_mech_frame}. 
We define the corresponding notions of Bayes-Nash equilibrium (see~\eqref{eq: def_med_Bayes_Nash_eq}), dominant equilibrium (see~\eqref{eq: med_dominant_strat}), and belief-dominant equilibrium (see~\eqref{eq: med_belief_dominant_strat}) for mediated mechanisms.
In Theorem~\ref{thm: rev_prin_1}, we
recover the revelation principle under certain settings (see table~\ref{tab: rev_prin}).


We conclude in section~\ref{sec: discussion} with some remarks and directions for future work.


\section{Mechanism Design Framework and the Revelation Principle}
\label{sec: preliminaries}

\subsection{Notational conventions}
Let $\1\{\cdot\}$ denote the indicator function that is equal to one if the predicate inside the brackets is true and is zero otherwise.
Let $\Delta(\cdot)$ denote the set of all probability distributions over the (finite) set within the parentheses.
Let $\supp(\cdot)$ denote the support of the probability distribution within the parentheses.
Let $\co (\cdot)$ denote the convex hull of the set within the parentheses.
For a function $f: X \to \Delta(Y)$, let $f(y|x) = f(x)(y)$ denote the probability of $y$ under the probability distribution $f(x)$.
Let 
\[
	L = \{(p_1,z_1); \dots;(p_t, z_t)\}
\]
denote a lottery with outcomes $z_j, 1 \leq j \leq t$, with their corresponding probabilities given by $p_j$.
We assume the lottery to be exhaustive (i.e. $\sum_{j = 1}^t p_j = 1$).
Note that we are allowed to have $p_j = 0$ for some values of $j$ and we can have $z_k = z_l$ even when $k \neq l$.
If a lottery $L$ consists of a unique outcome $z$ that occurs with probability $1$, then with an abuse of notation we will denote the lottery $L = \{(1,z)\}$ simply by $L = z$.
Similarly, if a probability distribution $f(x)$ assigns probability $1$ to $y$, then again with an abuse of notation we will write $f(x) = y$.
If, for each $x$, $f(x)$ has a singleton support, then with an abuse of notation we will treat $f$ as a function from $X$ to $Y$.

\subsection{Preliminaries}
\label{subsec: setup}
Let $[n] := \{1, 2, \dots, n\}$ be the set of \emph{players} participating in the 
system. 
Let $A$ denote the set of 
\emph{allocations} for this system.
We assume unless stated otherwise that the set of allocations is finite, say $A := \{\alpha^1, \dots, \alpha^l\}$.
For example, in the sale of a single item (or multiple items), it could represent the allocation of the item(s) to the different individuals.
In a routing system, such as traffic routing or internet packet routing, it could represent the different routing alternatives. 
More generally, in a resource allocation setting it could represent the assignment of resources to the participating agents (with their corresponding payments) that respect the system (and budget) constraints.
In a voting scenario, it could represent the winning candidate.
Thus, we imagine the allocations $\alpha \in A$ as being the various alternatives available to the 
system operator 
to implement.

Traditionally, each player is assumed to have a value for each of the allocations, and this defines the type of this player.
It describes the preferences of a player over the allocations, and further, by assuming EUT behavior, we get her preferences over the lotteries over these allocations.
Here, instead, we assume that
for each player $i \in [n],$ we have a finite set of \emph{outcomes} $\Gamma_i := \{\gamma_i^1, \dots, \gamma_i^{k_i}\}$, 
and player $i$'s 
\emph{type} 
is defined by her CPT preferences over the lotteries on this set $\Gamma_i$.
We imagine the set $\Gamma_i$ to capture the outcome features that are relevant to player $i$.
Thus the outcome set $\Gamma_i$ allows us to separate out the features that affect player $i$ from the underlying allocations that give rise to these outcomes.
We capture this relation between the allocation set and the outcome sets through a mapping $\zeta: A \to \Delta(\Gamma)$ that we call the \emph{allocation-outcome mapping}, where $\Gamma := \prod_i \Gamma_i$.
Let $\zeta_i: A \to \Delta(\Gamma_i)$ denote 
allocation-outcome mapping 
for player $i$ given by the marginal of $\zeta$ on the set $\Gamma_i$.
%

From 
a behavioral point of view it is natural to model a player's preferences on the outcome set $\Gamma_i$ rather than the allocation set $A$.
Then why is it that the sets $\Gamma_i$ and 
the mapping $\zeta$ are usually missing from the mechanism design framework prevalent in the literature?
At the end of subsection~\ref{subsec: mechframe}, after setting up the relevant notation, we will show that 
under EUT, from the point of view of the typical goals of the mechanism designer, it is enough to consider a transformation of the system where $\Gamma_i = A$, for all $i$, and the allocation-outcome mappings are \emph{trivial}, namely, $\zeta_i(\alpha) = \alpha$, for all $\alpha \in A, i \in [n]$
(this is shown formally in Appendix~\ref{sec: new_appendix}).
We will also show that this 
does not hold 
in general
when the players do not have EUT preferences, and in particular when they have CPT preferences. 

We model the preference behavior of the players using 
cumulative prospect theory 
(CPT) that we describe now. (For more details, see \citep{wakker2010prospect}.)

\subsection{CPT preference model}
\label{subsec: CPTpref}

Suppose $\Gamma_i$ is the outcome set for player $i$, who is associated with a \emph{value function} $v_i : \Gamma_i \to \bbR$ and two \emph{probability weighting functions} $w_i^\pm:[0,1] \to [0,1]$.
The value function $v_i$ partitions the set of outcomes $\Gamma_i$ into two parts: \emph{gains} and \emph{losses}; an outcome $\gamma_i \in \Gamma_i$ is said to be a gain if $v_i(\gamma_i) \geq 0$, and a loss otherwise.
The probability weighting functions $w_i^+$ and $w_i^-$ 
will be used for 
gains and losses, respectively. 
The probability weighting functions $w_i^\pm$ are assumed to satisfy the following:
\begin{inparaenum}[(i)]
	\item they are strictly increasing,
	\item $w_i^\pm(0) = 0$ and $w_i^\pm(1) = 1$.
\end{inparaenum}
We say that $(v_i, w_i^\pm)$ are the CPT features of 
player $i$.

Suppose 
player $i$ 
faces a \emph{lottery} $L_i \in \Delta(\Gamma_i)$
given by $\{(p_i^j,\gamma_i^j)\}_{1 \leq j \leq k_i}$.
Let $p_i := (p_i^1, \dots, p_i^{k_i})$.
Let $\beta_i := (\beta_i^1,\dots,\beta_i^{k_i})$ be a permutation of $(1,\dots,k_i)$ such that
\begin{equation}\label{eq: order}
	v_i(\gamma_i^{\beta_i^1}) \geq v_i(\gamma_i^{\beta_i^2}) \geq \dots \geq v_i(\gamma_i^{\beta_i^{k_i}}).
\end{equation}
Let $0 \leq j_r \leq k_i$ be such that $v_i(\gamma_i^{\beta_i^j}) \geq 0$, for $1 \leq j \leq j_r$, and $v_i(\gamma_i^{\beta_i^j}) < 0$, for $j_r < j \leq k_i$. 
(Here $j_r = 0$ when $v_i(\gamma_i^{\beta_i^j}) < 0$ for all $1 \leq j \leq k_i$.) 
The \emph{CPT value} $V_i(L_i)$ of the lottery $L_i$ is evaluated using the value function 
$v_i$ 
and the probability weighting functions 
$w_i^{\pm}$ 
as follows:
\begin{equation}\label{eq: CPT_value_discrete}
	V_i(L_i) := \sum_{j=1}^{j_{r}} \pi_i^+(j, p_i,\beta_i) v_i(\gamma_i^{\beta_i^j}) + \sum_{j=j_r+1}^{k_i} \pi_i^-(j, p_i,\beta_i) v_i(\gamma_i^{\beta_i^j}),
\end{equation}
where 
$\pi_i^+(j, p_i,\beta_i),1 \leq j \leq j_{r}$ and  $\pi_i^-(j, p_i,\beta_i), j_r < j \leq k_i$, 
are \emph{decision weights} defined via:
\begin{align*}
	\pi^+_{i}(1, p_i,\beta_i) &:= w_i^+(p_i^{\beta_i^1}),\\ 
	\pi_i^+(j, p_i,\beta_i) &:= w_i^+(p_i^{\beta_i^1} + \dots + p_i^{\beta_i^{j}}) - w_i^+(p_i^{\beta_i^1} + \dots + p_i^{\beta_i^{j-1}}), &\text{ for } &1 < j \leq k_i, \\
	 \pi_i^-(j, p_i,\beta_i) &:= w_i^-(p_i^{\beta_i^{k_i}} + \dots + p_i^{\beta_i^j}) - w_i^-(p_i^{\beta_i^{k_i}} + \dots + p_i^{\beta_i^{j+1}}), &\text{ for } &1 \leq j < k_i,\\
	 \pi^-_{i}(k_i, p_i,\beta_i) &:= w_i^-(p_i^{\beta_i^{k_i}}). 
\end{align*}
Although the expression on the right in equation~(\ref{eq: CPT_value_discrete}) depends on the permutation $\beta_i$, one can check that the formula evaluates to the same value $V_i(L_i)$ as long as the permutation $\beta_i$ satisfies~(\ref{eq: order}). 
The CPT value in equation~(\ref{eq: CPT_value_discrete}) can equivalently be written as:
\begin{align}\label{eq: CPT_value_cumulative}
 	V_i(L_i) &= \sum_{j = 1}^{j_r - 1} w_i^+\l(\sum_{i = 1}^j p_i^{\beta_i^j}\r)\l[v_i(\gamma_i^{\beta_i^j}) - v_i(\gamma_i^{\beta_i^{j+1}})\r] \nonumber\\
 	&+ w_i^+\l(\sum_{i = 1}^{j_r} p_i^{\beta_i^j}\r)v_i\l(\gamma_i^{\beta_i^{j_r}}\r) + w_i^-\l(\sum_{i = j_r + 1}^{k_i} p_i^{\beta_i^j}\r)v_i(\gamma_i^{\beta_i^{j_r+1}}) \nonumber \\
 	&+ \sum_{j = j_r + 1}^{k_i-1} w_i^-\l(\sum_{i = j+1}^{k_i} p_i^{\beta_i^j}\r)\l[v_i(\gamma_i^{\beta_i^{j+1}}) - v_i(\gamma_i^{\beta_i^{j}})\r].
 \end{align}

A person is said to have CPT preferences if, given a choice between lottery $L_i$ and lottery $L_i'$, she chooses the one with higher CPT value.

The \emph{expected utility} of player $i$ is completely characterized by her value function $v_i: \Gamma_i \to \bbR$.
For a lottery $L_i = \{(p_i^j, \gamma_i^j)\}_{1 \leq j \leq k_i}$, the expected utility is given by 
$$U_i(L_i) := \sum_{j = 1}^{k_i} p_i^j v_i(\gamma_i^j).$$
A person is said to have EUT preferences if, given a choice between lottery $L_i$ and lottery $L_i'$, she chooses the one with higher expected utility.
Observe that, if the probability weighting functions are linear, i.e. $w_i^\pm(p) = p$, for $p \in [0,1]$, then $V_i(L_i) = U_i(L_i)$ for all lotteries $L_i \in \Delta(\Gamma_i)$.
Thus EUT preferences are a special case of CPT preferences.



\subsection{Mechanism design framework}
\label{subsec: mechframe}

For each $i$, 
let $\Type_i$ denote the set from which 
the
permissible types for player $i$ are drawn.
Corresponding to any type $\type_i$ for player $i$,
let  
$v_i: \Gamma_i \to \bbR$ be her \emph{value function}, and
$w_i^{\pm}: [0,1] \to [0,1]$ be her \emph{probability weighting functions}.
Let $V_i^{\type_i}(L_i)$ denote the CPT value of the lottery $L_i \in \Delta(\Gamma_i)$ for player $i$ having type $\type_i$.
Thus, the type $\type_i$ completely determines the preferences of player $i$ over lotteries on her outcome set $\Gamma_i$.%
\footnote{
	Since we have assumed that the type of a player completely determines her CPT features, we are implicitly assuming \emph{private preferences}, i.e. the preference over lotteries on the outcome set for each player is her private information and does not depend on other players' information or types, also known as \emph{informational externalities} (see \citet{williams2008communication}).
}
We will assume 
that the sets $\Type_i$ are finite for all $i$.


Let $\type := (\type_1, \dots, \type_n)$ denote the profile of types of 
the 
players, and let $\Type := \prod_i \Type_i$.
We assume that each player knows her type but 
cannot observe the types of her opponents.

Let the set of players $[n]$, their corresponding type sets $\Type_i, i \in [n]$, 
the allocation set $A$, and the outcome spaces $\Gamma_i, i \in [n],$ together with the mapping $\zeta$ form an \emph{environment}, denoted by
\begin{equation}
\label{eq: def_environment}
	\cal{E} := \l([n], (\Type_i)_{i \in [n]}, A, (\Gamma_i)_{i \in [n]}, \zeta\r).
\end{equation}

A \emph{social choice function}
\begin{equation}
	\label{eq: def_scf}
	g: \Type \to \Delta(\Gamma)
\end{equation}
determines a 
lottery over the 
product of the outcome sets of the individual players 
given the type profile $\type$ of all the players.
The 
\emph{outcome choice function} for player $i$
corresponding to the social choice function $g$ is 
\begin{equation}
\label{eq: def_ocf}
	g_i: \Type \to \Delta(\Gamma_i),
\end{equation}
given by the restriction of $g$ to the set $\Gamma_i$,
and
represents the lottery faced by player $i$ given the type profile $\type$ of all the players.
We will treat the social choice function $g$ as the goal of the mechanism designer, i.e, 
the goal is 
to design a mechanism to implement a social choice function $g$
without having knowledge 
of the true types of the players.

Let an \emph{allocation choice function}
\begin{equation}
\label{eq: def_acf}
	f: \Type \to \Delta(A)
\end{equation}
represent the choice of the allocation to be implemented by the system operator given a type profile $\type \in \Type$.
Note that $f(\type)$ is a probability distribution over the allocations $A$.
Thus we allow the system operator to implement a randomized allocation.
A \emph{deterministic allocation choice function} maps each type profile to a unique allocation.
Since the mapping $\zeta$ is fixed and a part of the environment description, the allocation choice function $f$ effectively captures the goal of a mechanism designer.
More precisely, let $\cal{F}(g)$ denote the set of all allocation choice functions that induce the social choice function $g$, 
i.e. for all $\type \in \Type$, $g(\type)$ is the mixture probability distribution of the probability distributions $(\zeta(\alpha), \alpha \in A)$ with weights $f(\alpha|\type)$.
We note that the set $\cal{F}(g)$ is non-empty if and only if
\[
	g(\type) \in \co \{\zeta(\alpha): \alpha \in A \},
\]
for all $\type \in \Type$.
We wish to design a mechanism that would implement an allocation choice function in $\cal{F}(g)$.
Thus a social choice function is implementable if and only if we can implement an allocation choice function $f$ that satisfies
\begin{equation}
\label{eq: acf_to_scf}
	g(\gamma| \type) = \sum_{\alpha \in A} f(\alpha| \type) \zeta(\gamma| \alpha),
\end{equation}
for all $\gamma \in \Gamma, \type \in \Type$.
%
This raises 
the main question in mechanism design, 
namely whether we can design a game that results in the implementation of some given allocation choice function $f$ under certain rationality conditions on the players even when the system operator cannot observe the players' types. 

First, let us look at the 
the relationship between lotteries on allocations and lotteries on the outcome set of a given player. 
Any lottery $\mu \in \Delta(A)$ induces 
a lottery $L_i(\mu) \in \Delta(\Gamma_i)$ given by
\begin{equation}
\label{eq: all_lottery_to_outcome_lottery}
	L_i(\gamma_i|\mu) := \sum_{\alpha \in A} \mu(\alpha) \zeta_i(\gamma_i | \alpha).
\end{equation}
Given that player $i$ has type $\type_i$, we know that the CPT value of lottery $L_i(\mu)$ is $V_i^{\type_i}(L_i(\mu))$.
This induces a value for player $i$ with type $\type_i$ on the lottery $\mu$ denoted by
\begin{equation}
	\label{eq: def_W}
	W_i^{\type_i}(\mu) := V_i^{\type_i}(L_i(\mu)).
\end{equation}
This defines a \emph{utility function} $W_i^{\type_i} : \Delta(A) \to \bbR$ that gives the preference relation over the lotteries $\mu \in \Delta(A)$ for a player $i$ 
having type $\type_i$.
Let
\begin{equation}
\label{eq: def_u_of_alpha}
 	u_i^{\type_i}(\alpha) := V_i^{\type_i}(\zeta_i(\alpha)) = W_i^{\type_i}(\alpha)
 \end{equation} 
 be the CPT value of the lottery for player $i$ corresponding to allocation $\alpha$.
\footnote{
Even if $u_i^{\type_i} = u_i^{\tilde \type_i}$ for some $\type_i \neq \tilde \type_i$ it is 
sometimes
convenient to retain the connection to the underlying type.
	Notice that we have allowed different types of player $i$ to have the same CPT features.
	Later, when we discuss mechanism design with a common prior, which is a distribution on the types of all the players, it will let us differentiate between the types of players that have identical CPT features but distinct beliefs on the opponents' types.
	Mechanism design often focuses on ``naive type sets'', that is, the type set $\Type_i$ for each player $i$ is assumed 
	to be comprised of 
	exactly one element for each ``preference type'' of the player.
	Here, by preference type of a player we mean the preferences of the player on her outcome set.
	We borrow the expression ``naive type sets'' from \cite{borgers2011common}.
	In this paper, we do not assume the type sets to be naive.
	Such an assumption 
	would entail a 
	bijective correspondence between the types $\type_i$ and the CPT features $(v_i, w_i^\pm)$ for each player $i$.
	This distinction is relevant because besides having a preference type, a player can also have a ``belief type''.
	For example, the prior $F$ could be such that $F_{-i}(\type_i) \neq F_{-i}(\tilde \type_i)$ even when the value function and the probability weighting functions corresponding to the types $\type_i$ and $\tilde \type_i$ coincide. 
	(For more on this, see~\citet{bergemann2005robust,liu2009redundant}, and \citet[Chapter~10]{borgers2015introduction}.)
}
If player $i$ has EUT preferences, then we have that
\begin{equation}
\label{eq: W_EUT_linear}
	W_i^{\type_i}(\mu) = \sum_{\alpha \in A} \mu(\alpha) u_i^{\type_i}(\alpha).
\end{equation}

We now consider
a \emph{mechanism} 
\begin{equation}
\label{def: mechanism}
	\cM_0 := ((\Sig_i)_{i \in [n]}, h_0),
\end{equation}
consisting of a collection of finite \emph{signal sets} 
$\Sig_i$,
 one for each player $i$, 
 and an \emph{allocation function} 
 \begin{equation}
\label{def: allocation_function}
 	h_0: \Sig \to \Delta(A),
 \end{equation}
where $\Sig := \prod_{i \in [n]} \Sig_i$.
Note that the allocation function is allowed to be randomized.
Let $\sig_i$ denote
a typical element of $\Sig_i$, and $\sig := (\sig_i)_{i \in [n]}$ 
denote a typical element of $\Sig$, 
called a \emph{signal profile}.

It is straightforward to incorporate the feature that the outcome sets $\Gamma_i$
might be different from the allocation set $A$, and 
the corresponding allocation-outcome mapping $\zeta$, 
so as to 
extend the definition of a Bayes-Nash equilibrium strategy profile for the mechanism $\cM_0$ and the implementability of 
an
allocation choice function $f$ in Bayes-Nash equilibrium.
To do this, assume that 
the types of the individual players 
are drawn according to a prior distribution $F \in \Delta(\Type)$ and that this distribution is common knowledge among the agents 
and the system operator. 
Let $F_i \in \Delta(\Type_i)$ denote the marginal of $F$ on $\Type_i$.
Suppose player $i$ has type $\type_i$.
Then the belief of player $i$ about the types of other players is given by the conditional distribution 
$$F_{-i}(\type_{-i}|\type_i) := \frac{F(\type_i, \type_{-i})}{F_i(\type_i)}, \text{ for all } \type_{-i} \in \Type_{-i}, \type_i \in \supp F_i,$$
where $\type_{-i} := (\type_j)_{j \neq i}$ is the profile of types of all players other 
than 
player $i$, $\Type_{-i} := \prod_{j \neq i} \Type_j$.

{
Recall that $\sig_i$ denotes
a typical element of $\Sig_i$, and $\sig := (\sig_i)_{i \in [n]}$ 
denotes
a typical element of $\Sig$.
Let $\Sig_{-i} := \prod_{j \neq i} \Sig_j$ with a typical element denoted by $\sig_{-i}$.
Let 
\begin{equation}
\label{eq: strategy}
	\sigma_i: \Type_i \to \Delta(\Sig_i)
\end{equation}
be a \emph{strategy} for player $i$, and let $\sigma := (\sigma_1, \sigma_2, \dots, \sigma_n)$ denote a strategy profile.
Let $\sigma_{-i} := (\sigma_j)_{j \neq i}$ denote the strategy profile of all players other than player $i$.
For any type $\type_i$ (such that $F_i(\type_i) > 0$) and signal $\sig_i$,
consider the probability distribution 
$\mu_i(\type_i, \sig_i; \cM_0, F, \sigma_{-i}) \in \Delta(A)$ 
given by 
\begin{equation}
\label{eq: belief_bayes_induced}
	\mu_i(\alpha | \type_i, \sig_i; \cM_0, F, \sigma_{-i}) 
	:= \sum_{\type_{-i} \in \Type_{-i}} F_{-i}(\type_{-i}|\type_i) \sum_{\sig_{-i} \in \Sig_{-i}} \prod_{j \neq i} \sigma_j(\sig_j | \type_j) h_0(\alpha | \sig),
\end{equation}
for all  $\alpha \in A$.
Suppose player $i$ has type $\type_i$ (such that $F_i(\type_i) > 0$),
and she chooses to signal $\sig_i$. Then, by Bayes' rule, the lottery faced by player $i$ is given by 
$$L_i\l(\mu_i(\type_i, \sig_i; \cM_0, F, \sigma_{-i})\r).$$
This comes from the assumption that player $i$ knows her type $\type_i$, the common prior $F$, the strategies $\sigma_j, j \neq i$ of her opponents, and the mapping $\zeta_i$.
Given that player $i$ has type $\type_i$, her CPT value for the lottery 
$L_i(\mu_i(\type_i, \sig_i; \cM_0, F, \sigma_{-i})$
is given by 
$$W_i^{\type_i}(\mu_i(\type_i, \sig_i; \cM_0, F, \sigma_{-i})) = V_i^{\type_i}\l(L_i(\mu_i(\type_i, \sig_i; \cM_0, F, \sigma_{-i})\r),$$
where we recall that $W_i^{\type_i}(\mu)$ is the CPT value of player $i$ with type $\type$ for the lottery $L_i(\mu) \in \Delta(\Gamma_i)$ induced by the distribution $\mu \in \Delta(A)$. 
Let the \emph{best response strategy set} $BR_i(\sigma_{-i})$ for player $i$ to a strategy profile $\sigma_{-i}$ of her opponents consist of 
all strategies $\sigma_i^*: \Type_i \to \Delta(\Sig_i)$ such that 
\begin{align}
\label{eq: Bayes_Nash_W_def}
	W_i^{\type_i}(\mu_i(\type_i, \sig_i; \cM_0, F, \sigma_{-i})) \geq W_i^{\type_i}(\mu_i(\type_i, \sig_i'; \cM_0, F, \sigma_{-i})),
\end{align}
for all $\type_i \in \supp F_i, \sig_i \in \supp \sigma_i^*(\type_i), \sig_i' \in \Sig_i$.
The rationale behind this definition is that a player's best response strategy $\sigma^*$ should not assign positive probability to any suboptimal signal $\sig_i$ given her type $\type_i$.

A strategy profile $\sigma^*$ is said to be an \emph{$F$-Bayes-Nash equilibrium} for the environment $\cal{E}$ and common prior $F$ with respect to the mechanism $\cM_0$ if, for each player $i$, we have
\begin{equation}
\label{eq: def_BNE}
	\sigma_i^* \in BR_i(\sigma^*_{-i}).
\end{equation}
We will refer to $\sigma^*$ simply as a Bayes-Nash equilibrium when the respective environment $\cal{E}$, the common prior $F$, and mechanism $\cM_0$ are clear from the context.


We say that the allocation choice function $f$ is 
\emph{implementable in $F$-Bayes-Nash equilibrium} 
by a mechanism
if there exists a mechanism $\cM_0$ and an $F$-Bayes-Nash equilibrium $\sigma$ 
such that $f$ is the induced distribution, i.e. 
for all $\type_i \in \supp F_i, \alpha \in A$, 
we have
\begin{equation}
\label{eq: def_implementable}
	f(\alpha|\type) = \sum_{\sig \in \Sig} \l(\prod_{i \in [n]} \sigma_i(\sig_i|\type_i)\r) h_0(\alpha| \sig).
\end{equation}
}

An alternative notion is that of 
an allocation choice function $f$ being \emph{implementable in dominant equilibrium}. 
The traditional notion states that a strategy $\sigma_i$ is a \emph{dominant strategy} for player $i$ if the signals in the support of $\sigma_i(\type_i)$ are optimal given player $i$'s type $\type_i$ 
and any signal profile $\sig_{-i}$ of the opponents.
More precisely, if we let
\begin{align}
	\mu_i(\type_i, \sig_i; \cM_0, \sig_{-i}) := h_0(\sig_i, \sig_{-i}),
\end{align}
then $\sigma_i^*$ is a dominant strategy if, for all $\type_i \in \type_i$, $\sig_i \in \supp \sigma_i^*(\type_i)$, 
$\sig_i' \in \Sig_i$, 
and $\sig_{-i} \in \Sig_{-i}$, we have
\begin{equation}
\label{eq: dominant_strat}
	W_i^{\type_i}(\mu_i(\type_i, \sig_i; \cM_0, \sig_{-i})) \geq W_i^{\type_i}(\mu_i(\type_i, \sig_i'; \cM_0, \sig_{-i})).
\end{equation}
Thus, if player $i$ employs a dominant strategy, then regardless of the signal profile of the opponents she always signals a best response given her type.
A dominant equilibrium is one in which each player plays a dominant strategy.
We say that an allocation choice function $f$ is implementable in dominant equilibrium 
if there exists a mechanism $\cM_0$ and a strategy profile $\sigma^* = (\sigma_1^*, \ldots, \sigma_n^*)$
where each $\sigma_i^*$ is a dominant strategy (equivalently, $\sigma^*$ is a dominant equilibrium)
such that \eqref{eq: def_implementable} holds for all $\type_i \in \Theta, \alpha \in A$.

Under EUT, if a signal $\sig_i$ is a best response of player $i$ for all of the opponents' signal profiles, then it is also a best response for any belief $G_{-i} \in \Delta(\Sig_{-i})$ of player $i$ over her opponents' signal profiles.
However, under CPT, this need not hold.
(See example~\ref{ex: not_belief_dominant_VCG}.)
This observation leads us to 
the following stricter notion of dominant strategies under CPT.
We call a strategy $\sigma_i$ a
\emph{belief-dominant strategy} for player $i$ if the signals in the support of $\sigma_i(\type_i)$ are optimal given player $i$'s type $\type_i$ and any belief $G_{-i} \in \Delta(\Sig_{-i})$ 
she has
on the signal profile of 
her 
opponents.
If we let
\begin{align}
	\mu_i(\type_i, \sig_i; \cM_0, G_{-i}) := \sum_{\sig_{-i}} G_{-i}(\sig_{-i}) h_0(\sig_i, \sig_{-i}),
\end{align}
then $\sig_i^*$ is a belief-dominant strategy 
for player $i$
if, for all $\type_i \in \type_i$, $\sig_i \in \supp \sigma_i^*(\type_i)$, 
$\sig_i' \in \Sig_i$, 
and $G_{-i} \in \Delta(\Sig_{-i})$, we have
\begin{equation}
	\label{eq: belief_dominant_strat}
	W_i^{\type_i}(\mu_i(\type_i, \sig_i; \cM_0, G_{-i})) \geq W_i^{\type_i}(\mu_i(\type_i, \sig_i'; \cM_0, G_{-i})).
\end{equation}
It is straightforward to check that under EUT 
a strategy is dominant if and only if it is belief-dominant.
A \emph{belief-dominant equilibrium} is one in which every player plays a belief-dominant strategy.
We say that an allocation choice function $f$ is implementable in belief-dominant equilibrium 
if there exists a mechanism $\cM_0$ and a strategy profile $\sigma^* = (\sigma_1^*, \ldots, \sigma_n^*)$
where each $\sigma_i^*$ is a belief-dominant strategy (equivalently, $\sigma^*$ is a belief-dominant equilibrium)
such that \eqref{eq: def_implementable} holds for all $\type_i \in \Theta, \alpha \in A$.

Note that 
if $\sigma^*$ is a belief-dominant strategy profile, 
and thus a belief-dominant equilibrium,
then it is a dominant strategy profile,
i.e. a dominant equilibrium,
and also 
an $F$-Bayes-Nash equilibrium with respect to any prior distribution $F$
on type profiles.

\input{plots/wt_plot_prelec}

\begin{example}
\label{ex: not_belief_dominant_VCG}
	Let $n = 2$. Let $\Type_1 = \Type_2 = \{\Up, \D\}$. 
	Let $A = \{a, b, c\}$, $\Gamma_1 = \{\I, \II, \III, \IV, \V\}$, and $\Gamma_2 = A$.
	Let the allocation-outcome mapping be given by the product distribution of the marginals $\zeta_1$ and $\zeta_2$, given by
	$\zeta_1(a) = \{(1/2, \I); (1/2, \V) \}, \zeta_1(b) = \{(1/2, \II); (1/2, \IV) \}, \zeta_1(c) = \{(1, (\III)) \},$
and 
	$\zeta_2(\alpha) = \alpha, \forall \alpha \in A.$
Let the probability weighting functions for gains for the two players be given by
\[
	w_1^+(p) = \exp \{-(-\ln p)^{0.5}\}, w_2^+(p) = p,
\]
for $p \in [0,1]$.
In this example, we consider only lotteries with outcomes in the gains domain, and hence 
an arbitrary 
probability weighting function for the losses can be assumed.
Here, for player $1$'s probability weighting function, we employ the form 
suggested by \citet{prelec1998probability}
(see figure~\ref{fig: wt_pwf_prelec}).
Note that player $2$ has EUT preferences.
Let the value functions $v_1$ and $v_2$ be given by
\begin{center}
\begin{tabular}{| c | c | c | c | c | c |}
\hline
$v_1$ & I & II & III & IV & V\\
\hline
$\Up$ & $2x$ & $x+1$ & $1.99$ & $1$ & $0$\\ 
\hline
$\D$ & $0$ & $0$ & $1$ & $0$ & $0$\\ 
\hline
\end{tabular}
\hspace{1in}
\begin{tabular}{| c | c | c | c | c | c |}
\hline
$v_2$ & a & b & c\\
\hline
$\Up$ & $1$ & $0$ & $2$\\ 
\hline
$\D$ & $0$ & $1$ & $2$\\ 
\hline
\end{tabular}
\end{center}
where $x := 1/w_1^+(0.5) = 2.2992$.
Note that $2x = 4.5984$ and $x + 1 = 3.2992$.
We have,
\begin{align*}
	V_1^\Up(L_1(a)) &= 2x w_1^+(0.5) = 2, \\
	V_1^\Up(L_1(b)) &= 1 + x w_1^+(0.5) =2,
\end{align*}
and, 
\begin{align*}
	V_1^\Up(0.5L_1(a) + 0.5L_1(b)) &= w_1^+(0.75) + x w_1^+(0.5) + (x-1) w_1^+(0.25) \\
	&= 1.9851.
\end{align*}
(Here, we have $w_1^+(0.25) = 0.3081$, $w_1^+(0.5) = 0.4349$, and $w_1^+(0.75) = 0.5849$.)
Consider the mechanism $\cM = ((\Sig_i)_{i \in [n]}, h_0)$, where $\Sig_1 = \Sig_2 = \{\Up, \D\}$, and $h_0$ is given by
\begin{align*}
	h_0(\Up,\Up) = a, h_0(\Up,\D) = b, h_0(\D,\Up) = c, h_0(\D,\D) = c.
\end{align*}
Consider the strategies
\[
	\sigma_i(\Up) = \Up, \text{ and } \sigma_i(\D) = \D,
\]
for both the players $i$.
It is easy to see that these strategies are dominant for both the players.
However, if player $1$ has type $\Up$ and believes that there is an equal chance of player $2$ reporting her strategy to be $\Up$ and $\D$, then player $1$'s best response is to report $\D$.
Thus, $\sigma_1$ is not a belief-dominant strategy for player $1$.

\qed
\end{example}

We will now look at the remark made earlier about the absence of the distinction between the allocation set and the outcome sets in classical mechanism design, and why it is important to consider this distinction under CPT.
In Appendix~~\ref{sec: new_appendix}, we show that under EUT it suffices to consider the scenario where
the outcome set of each player is the same as the allocation set by the simple expedient of interpreting
each type $\type_i \in \Type_i$ in terms of the utility function on allocations that it defines via 
\eqref{eq: def_u_of_alpha}.

\black

While equation~\eqref{eq: W_EUT_linear} holds under EUT, under CPT in general it 
does not hold, and in general the utility function $W_i^{\type_i}$ is not completely determined by the values $u_i^{\type_i}(\alpha), \forall \alpha \in A$.
Thus, we can either characterize the type of a player by her utility function $W_i^{\type_i}$ or by her CPT features which, combined with the mapping $\zeta_i$, together 
define 
the utility function $W_i^{\type_i}$.
In any given setting, it is more natural to put behavioral assumptions on the CPT features $(v_i, w_i^\pm)$
than on the utility function $W_i^{\type_i}$.%
\footnote{
	Note that, in general, the preferences defined by the utility function $W_i^{\type}$ over the lotteries over the allocation set may not be given by CPT preferences directly, i.e. there need not exist any probability weighting functions $\tilde w_i^\pm$ such that, for all $\mu \in \Delta(A)$, $W_i^{\type_i}(\mu)$ is equal to the CPT value corresponding to the value function $u_i^{\type_i}$ on $A$ and the probability weighting functions $\tilde w_i^\pm$.
	To see this, consider a type $\type_i$ for player $i$ such that
	$V_i^{\type_i}(L_i') = V_i^{\type_i}(L_i'') > V_i^{\type_i}(0.5L_i' + 0.5L_i'')$, for lotteries $L_i', L_i'' \in \Delta(\Gamma_i)$.
	See \citet{phade2018learning} for an example of CPT preferences and
	lotteries (over $4$ outcomes) that satisfy the above condition.
	Let there be two allocations $\alpha'$ and $\alpha''$ such that $\zeta_i(\alpha') = L_i'$ and $\zeta_i(\alpha'') = L_i''$.
	If $W_i^{\type_i}$ were to correspond to any CPT preference directly on the allocation set then, by the first order stochastic dominance property of CPT, we would get $W_i^{\type_i}(0.5 \alpha' + 0.5 \alpha'') = W_i^{\type_i}(\alpha') = W_i^{\type_i}(\alpha'')$.
	But, since this is not true for the setting under consideration, we get that $W_i^{\type_i}$ cannot correspond to any CPT preference directly over $A$.
}
Hence, we include the sets $\Gamma_i$ and the mappings $\zeta_i$, for all $i$, in our system model under CPT.


\section{The revelation principle}
\label{sec: rev_prin}

A mechanism $\cM_0 = ((\Sig_i)_{i \in [n]}, h_0)$ is called a \emph{direct mechanism} if $\Sig_i = \Type_i,$ for all $i$.
Let 
$\cM^d_0 := ((\Type_i)_{i \in [n]}, h_0^d)$ 
denote a direct mechanism, where 
\begin{equation}
	\label{eq: def_direct_allocation_func}
	h_0^d: \Type \to \Delta(A)
\end{equation}
is the \emph{direct allocation function}.
Corresponding to a direct mechanism, let $\sigma^d_i : \Type_i \to \Type_i$ denote the truthful strategy for player $i$, given by
\begin{equation}
\label{eq: def_dir_strat}
	\sigma_i^d(\type_i) = \type_i,
\end{equation}
for all $\type_i \in \Type_i$.
An allocation choice function $f$ is said to be \emph{truthfully implementable} in $F$-Bayes-Nash equilibrium 
(resp. dominant equilibrium or belief-dominant equilibrium)
if there exists a direct mechanism $\cM^d_0$ such that the truthful strategy profile $\sigma^d$ 
is an $F$-Bayes-Nash equilibrium 
(resp. dominant equilibrium or belief-dominant equilibrium), 
and it induces $f$.

The \emph{revelation principle}%
\footnote{
	This is the version 
	of the revelation principle
	commonly referred to in the mechanism design context. 
	Another version of the revelation principle appears in the context of correlated equilibrium \citep{aumann1974subjectivity,aumann1987correlated}.
	This is concerned with 
	an $n$-player non-cooperative game in normal form.
	A mediator draws a message profile, comprised of a message for each player, from a fixed joint probability distribution on the set of message profiles, and sends each player her corresponding message.
	The joint distribution over message profiles used is assumed to be common knowledge between the mediator and
	all the players.
	Based on her received signal, each player chooses her action (possibly from a probability distribution over her action set).
	When the message set for each player is the same as her action set and the probability distribution on the set of message profiles (or equivalently action profiles) is such that truthful strategy, i.e. the strategy of choosing the action that is received as a message from the mediator, is a Nash equilibrium, then such a probability distribution is said to be a correlated equilibrium. 
	Under EUT, the set of all correlated equilibria of a game is characterized as the union over all possible message sets and mediator distributions, of the sets of joint distributions on the action profiles of all players, arising from all the Nash equilibria for the resulting game.
	See \citet{phade2018learning} for a discussion on the revelation principle for correlated equilibrium when players have CPT preferences.
	\citet{myerson1986multistage} has considered a further generalization to games with incomplete information in which each player first reports her type.
	Analyzing such settings under CPT would entail dynamic decision making and is beyond the scope of this paper.
} 
says that if an allocation choice function is implementable in Bayes-Nash equilibrium (resp. dominant equilibrium or belief-dominant equilibrium) by a mechanism, then it is also truthfully implementable in Bayes-Nash equilibrium (resp. dominant equilibrium or belief-dominant equilibrium) by a direct mechanism.
When the players are restricted to have EUT preferences and the outcome set of each player is assumed to 
be 
the same as the allocation set with the trivial allocation-outcome mapping, \citet{myerson1982optimal} proved that the revelation principle holds for both the versions - Bayes-Nash equilibrium and dominant equilibrium (and hence also for belief-dominant equilibrium, since dominant strategies are equivalent to belief-dominant strategies under EUT). 
It is easy to extend this result to the general setting where 
some of the individual outcome sets might differ from the allocation set,
provided the players are restricted to have EUT preferences.
Indeed, in Appendix~\ref{sec: new_appendix} it is proved that,
under EUT, an allocation choice function $f$ is implementable in $F$-Bayes-Nash (resp. dominant or belief-dominant) equilibrium by a mechanism $\cM_0$ for the environment $\cal{E}$ with the equilibrium strategy $\sigma$, if and only if, for the corresponding environment 
$\cal{E}'$ (defined in Appendix~\ref{sec: new_appendix}),
the corresponding allocation choice function $f'$ is implementable in $F'$-Bayes-Nash (resp. dominant or belief-dominant) by the same mechanism $\cM_0$ with the corresponding equilibrium strategy $\sigma'$.
We now observe that $\cM_0$ is a direct mechanism for environment $\cal{E}$ if and only if it is a direct mechanism for environment $\cal{E}'$.
Also, $\sigma_i$ is the truthful strategy with respect to the environment $\cal{E}$ and a direct mechanism $\cM_0$, if and only if, the corresponding strategy $\sigma_i'$ 
is the truthful strategy with respect to the environment $\cal{E}'$ and the same direct mechanism $\cM_0$.
These observations together give us the required revelation principle under EUT for the setting where 
the outcome sets of some of the players can differ from the allocation set.

The following example shows that the revelation principle need not hold when players have CPT preferences.
We will consider implementability in Bayes-Nash equilibrium in this example.

\begin{example}
\label{ex: no_rev_impl}
Let there be two players, i.e. $n = 2$. Let each player belong to one of the three types: Mildly Favorable ($\MF$), Unfavorable ($\UF$), and Super Favorable ($\SF$), i.e. $\Type_1 = \Type_2 = \{\MF, \UF, \SF\}$.
Let the outcome sets for both the players be $\Gamma_1 = \Gamma_2 = \{\I,\II,\III,\IV,\V\}$. 
Let the value functions $v_1$ and $v_2$ for both the players be as shown below.
\begin{center}
\begin{tabular}{| c | c | c | c | c | c |}
\hline
 & I & II & III & IV & V\\
\hline
$\MF$ & $13.616$ & $8.616$ & $5.816$ & $3.8$ & $0$\\ 
\hline
$\UF$ & $-190$ & $-100$ & $-1\K$ & $-50$ & $0$\\ 
\hline
$\SF$ & $0$ & $0$ & $1\M$ & $0$ & $0$\\
\hline
\end{tabular}
\end{center}
Observe that a player with type $\MF$ has mild gains for all the outcomes, a player with type $\UF$ has medium losses for all outcomes except outcome $\III$, where she has a big loss, and a player of type $\SF$ has a huge gain for outcome $\III$ and zero gains otherwise.

Let the probability weighting functions for both the players, for all of their types, be given by the following piecewise linear functions:
\begin{equation*}
	w^+(p) = \begin{cases}
		(8/7)p, &\text{ for } 0 \leq p < (7/32),\\
		(1/4) + (2/3)(p-7/32), &\text{ for } (7/32) \leq p < 25/32,\\
		(5/8) + (12/7)(p - 25/32), &\text{ for } (25/32) \leq p < 1,
	\end{cases}
\end{equation*}
for gains, and
\begin{equation*}
	w^-(p) = \begin{cases}
		(3/2)p, &\text{ for } 0 \leq p < (1/8),\\
		(3/16) + (1/2)(p-1/8), &\text{ for } (1/8) \leq p < 3/4,\\
		(1/2) + 2(p - 3/4), &\text{ for } (3/4) \leq p < 1,
	\end{cases}
\end{equation*}
for losses.
See figure~\ref{fig: ex_weight}.


%
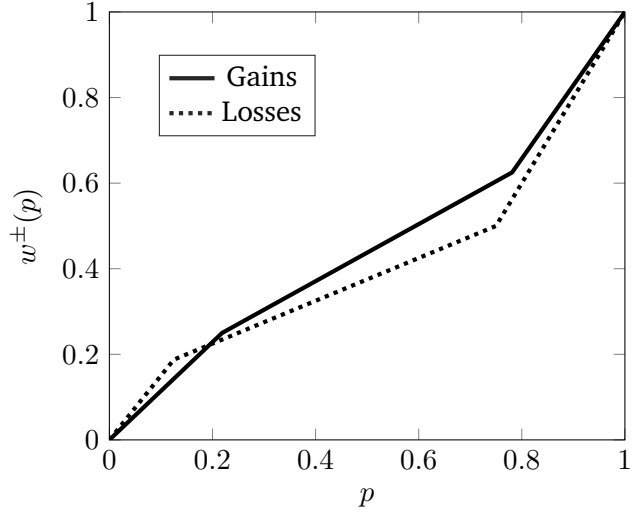
\begin{figure}
\centering
\begin{tikzpicture}[scale = 1]
    \begin{axis}[
    	xmin = 0,
    	xmax = 1,
    	ymin = 0,
    	ymax = 1,
        xlabel=$p$,
        ylabel=$w^\pm(p)$,
        every axis plot/.append style={ultra thick},
        legend style={ at={(0.4,0.9)}, anchor=north east, align=left, draw=white!15!black}]
\addplot []
  table[row sep=crcr]{%
	0   0\\
    0.2188   0.2500\\
    0.7812   0.6250\\
    1   1\\
};
\addlegendentry{Gains}

\addplot [dotted]
  table[row sep=crcr]{%
    0   0\\
    0.1250   0.1875\\
    0.75   0.5\\
    1   1\\
};
\addlegendentry{Losses}


\end{axis}
\end{tikzpicture}%
\caption{Probability weighting functions for the players in example~\ref{ex: no_rev_impl}.}
\label{fig: ex_weight}
\end{figure}

Let the prior distribution $F$ be such that the types of the players are independently sampled with probabilities, 
\begin{equation}\label{eq: ex_prior1}
	\bbP(\MF) = 1/2, \bbP(\UF) = 3/8, \bbP(\SF)  = 1/8.	
\end{equation}

Let $A = \{a, b, c\}$ be the allocation set, and let the allocation-outcome mapping be given by
\begin{align*}
	\zeta(a) &= \{(1/2, (\I,\I)); (1/2, (\V, \V)) \},\\
	\zeta(b) &= \{(1/2, (\II,\II)); (1/2, (\IV, \IV)) \},\\
	\zeta(c) &= (\III,\III).
\end{align*}

Consider the allocation choice function $f^*$ given by
\begin{align*}
	&f^*(\SF, \type_2) = f^*(\type_1, \SF) = c, \quad \forall \type_1 \in \Type_1, \type_2 \in \Type_2,\\
	&f^*(\UF, \type_2) = f^*(\type_1, \UF) = \{(1/2, a); (1/2, b)\}, \quad \forall \type_1 \in \{\MF, \UF \}, \type_2 \in \{\MF, \UF \},\\
	&f^*(\MF, \MF) = \{(1/2,a); (1/2,b)\}.
\end{align*}
We will now show that $f^*$ is \emph{not} truthfully implementable in $F$-Bayes-Nash equilibrium by a direct mechanism.
However, if we do not restrict ourselves to direct mechanisms, then we will show that it is possible to implement $f^*$ in $F$-Bayes-Nash equilibrium.
We will then conclude that the revelation principle does not hold 
for Bayes-Nash implementability
when the players have CPT preferences. 

We observe that if either of the players is of type $\SF$ then under 
the 
allocation $c$ the players with type $\SF$ get the maximum possible reward, 
i.e. $1\M$. 
This motivates implementing allocation $c$ if either of the players is of type $\SF$. 
Now suppose none of the players has type $\SF$.
If player $1$ is of type $\UF$, then  
in claim \ref{cl: ex_valuUF}, we show that player $1$'s CPT value for 
the 
lottery $L_i(\mu)$ corresponding to a distribution $\mu \in \Delta(A)$ is maximized when
\begin{equation}
\label{eq: mu_hf_hf}
	\mu = \{(1/2, a); (1/2, b); (0,c) \}.
\end{equation}
Thus, if at least one of the players has type $\UF$ and none of the players have type $\SF$, then the distribution in \eqref{eq: mu_hf_hf} gives the best CPT value for the players with type $\UF$.
This motivates the following definition:
we will call an allocation choice function $f$ \emph{special} if it satisfies 
\begin{equation}
\label{eq: ex_f_SF}
	f(\SF, \type_2) = f(\type_1, \SF) = \{(1, c)\}, \forall \type_1 \in \Type_1, \type_2 \in \Type_2,
\end{equation}
and
\begin{equation}
\label{eq: ex_f_UF}
	f(\UF, \type_2) = f(\type_1, \UF) = \{(1/2, a); (1/2, b)\}, \forall \type_1, \type_2 \in \{\MF,\UF\}.
\end{equation}
Note that $f^*$ is special.

After proving claim \ref{cl: ex_valuUF}, we will show that it is impossible to truthfully implement any special allocation choice function
in $F$-Bayes-Nash equilibrium by a direct mechanism.
In particular, this would imply that $f^*$ is not truthfully implementable by a direct mechanism.
We will then give a mechanism $\cal{M}_0$ that implements $f^*$ in $F$-Bayes-Nash equilibrium.

\begin{claim}
\label{cl: ex_valuUF}
 	The CPT value $V_1^{\UF}(L_1(\mu))$ is maximized when $\mu$ is given by \eqref{eq: mu_hf_hf}.
 \end{claim} 
\begin{proof}[Proof of Claim~\ref{cl: ex_valuUF}] 
	Consider a lottery
\[
	\mu = \{(x, a); (y, b); (z,c) \},
\]
where 
$x, y, z$ are nonnegative numbers with 
$x + y + z = 1$.
Then the outcome lottery for player $1$ is
\[
	L_1(\mu) = \{(x/2, \I); (y/2, \II); (z, \III); (y/2, \IV); (x/2, \V)\}.
\]
	CPT satisfies \emph{strict stochastic dominance} 
	\citep{chateauneuf1999axiomatization},
	i.e. shifting positive probability mass from an outcome to a strictly preferred outcome leads to a strictly preferred lottery. 
	This implies that $z = 0$ in the optimal solution.
	Taking $z = 0$ and $y = 1-x$, from \eqref{eq: CPT_value_cumulative},
	we have
	\begin{align*}
		E(x) &:= V_1^{\UF}(\{(x/2, \I); (1/2-x/2, \II); (0, \III); (1/2 - x/2, \IV); (x/2, \V)\}) \\
		&= - 90w^-(x/2) - 50w^-(1/2) - 50w^-(1 - x/2).
	\end{align*}
	We can verify that this function is maximized at $x = 1/2$. 
	See figure~\ref{fig: ex_valUF}.

\begin{figure}
\centering
\begin{tikzpicture}[scale = 1]
    \begin{axis}[
    	xmin = 0,
    	xmax = 1,
    	ymin = -75,
    	ymax = -65,
        xlabel=$x$,
        ylabel=$E(x)$,
        every axis plot/.append style={ultra thick},
        ]
\addplot []
  table[row sep=crcr]{%
	0   -68.7500\\
    0.25   -73.1250\\
    0.5   -66.25\\
    1   -71.2500\\
};


\end{axis}
\end{tikzpicture}%
\caption{Plot of expression $E(x)$ in Claim~\ref{cl: ex_valuUF}.}
\label{fig: ex_valUF}
\end{figure}
\end{proof}


Suppose we have a direct mechanism $\cM_0^d = h_0^d$ that truthfully implements a special allocation choice function $f$.
Then the allocation function $h_0^d$ must be equal to the allocation choice function $f$.
Since $f$ satisfies \eqref{eq: ex_f_SF} and \eqref{eq: ex_f_UF},
the only freedom left is in the choice of $f(\MF, \MF)$. 
Let
\[
	h_0^d(\MF, \MF) = f(\MF, \MF) = \{(x',a); (y',b); (z',c) \},
\]
where 
$x', y', z'$ are nonnegative numbers with 
$x' + y' + z' = 1$.
The lottery faced by a player of type $\MF$ signaling truthfully would then be
\begin{align*}
	L_1(\mu_{1}(\MF, \MF; \cM_0^d, F, \sigma^d_{-i})) =  \{(3/32 + &x'/4, \I);  (3/32 + y'/4, \II); \\
	&  (1/8 + z'/2, \III); (3/32 + y'/4, \IV); (3/32 + x'/4, \V)\}.
\end{align*}
We obtain this by using the belief $F_{-1}(\cdot| \MF)$ of player $1$ on the type of player $2$ given by \eqref{eq: ex_prior1}, the truthful strategy $\sigma_2^d$ for player $2$, and the allocation function $h_0^d$ in \eqref{eq: belief_bayes_induced}. 

\begin{claim}
\label{cl: ex_val_MF}
	For any 
	nonnegative 
	$x', y', z'$ such that $x' + y' + z' = 1$, we have
	$$V_1^{\MF}(L_1(\mu_{1}(\MF, \MF; \cM_0^d, F, \sigma^d_{-i}))) < 5.816.$$
\end{claim}
\begin{proof}[Proof of Claim~\ref{cl: ex_val_MF}]
	We have
	\begin{align*}
		V_1^{\MF}L_1(\mu_{1}(\MF, \MF; \cM_0^d, F, \sigma^d_{-i})) = &3.8w^+(29/32 - x'/4) + 2.016 w^+(18/32 + z'/4)\\
		 &+ 2.8w^+(14/32 - z'/4) + 5w^+(3/32 + x'/4).
	\end{align*}
	We observe that the expression,
	\[
		E_1(z') := 2.016 w^+(18/32 + z'/4) + 2.8w^+(14/32 - z'/4),
	\]
	is maximized at $z' = 0$ with value $E_1(0) = 2.0743$.
	See figure~\ref{fig: ex_valMFZ}.

\begin{figure}
\centering
\begin{tikzpicture}[scale = 1]
    \begin{axis}[
    	xmin = 0,
    	xmax = 1,
    	ymin = 1.95,
    	ymax = 2.10,
        xlabel=$z'$,
        ylabel=$E_1(z')$,
        every axis plot/.append style={ultra thick},
        ]
\addplot []
  table[row sep=crcr]{%
	0   2.0743\\
    0.875   1.9600\\
    1   1.9680\\
};


\end{axis}
\end{tikzpicture}%
\caption{Plot of expression $E_1(z')$ in Claim~\ref{cl: ex_val_MF}.}
\label{fig: ex_valMFZ}
\end{figure}
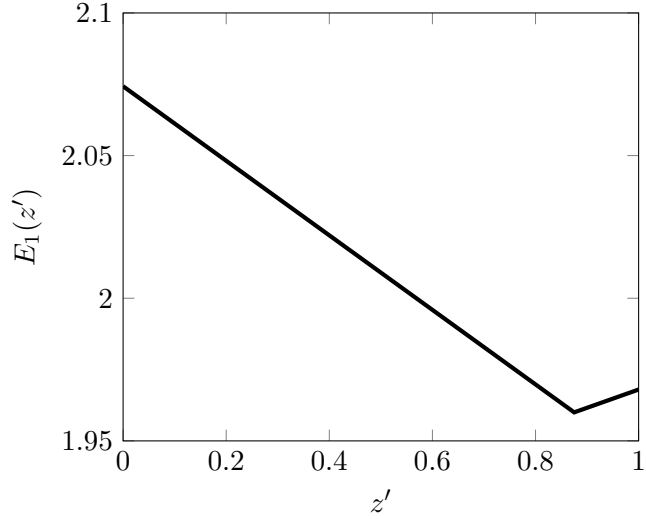
	We can therefore set $z' = 0$, since this choice would also lead to the least constrained problem of maximizing the expression
	\[
		E_2(x') := 3.8w^+(29/32 - x'/4) + 5w^+(3/32 + x'/4) + 2.0743,
	\]
	which we can see is maximized at $x' = 0$ and $x' = 1$.
	At $z'= 0$, and either $x' = 0$ or $x' =1$, 
	we have $V_1^{\MF}L_1(\mu_{1}(\MF, \MF)) = 5.7993$.
	See figure~\ref{fig: ex_valMFX}.

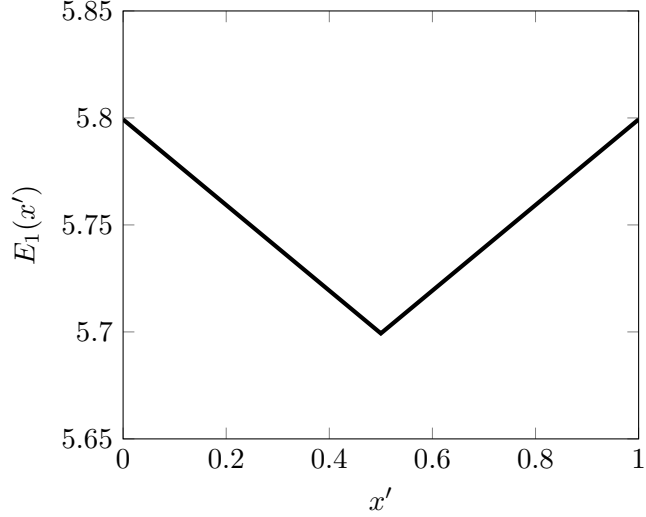
\begin{figure}
\centering
\begin{tikzpicture}[scale = 1]
    \begin{axis}[
    	xmin = 0,
    	xmax = 1,
    	ymin = 5.65,
    	ymax = 5.85,
        xlabel=$x'$,
        ylabel=$E_1(x')$,
        every axis plot/.append style={ultra thick},
        ]
\addplot []
  table[row sep=crcr]{%
	0   5.7993\\
    0.5   5.6993\\
    1   5.7993\\
};


\end{axis}
\end{tikzpicture}%
\caption{Plot of expression $E_2(x')$ in Claim~\ref{cl: ex_val_MF}.}
\label{fig: ex_valMFX}
\end{figure}
	This establishes the claim. 
\end{proof}

Thus, player $1$ will defect from 
the 
truthful strategy and report $\SF$ when her true type is $\MF$,
because if she does so the allocation $c$ will be implemented by the system operator, which results in her outcome being $\III$, hence giving her a value of $5.816$. 
Hence, truthful strategies do not form a Bayes-Nash equilibrium under the allocation function $h_0^d$. 
And hence, any allocation choice function $f$ that satisfies \eqref{eq: ex_f_SF} and \eqref{eq: ex_f_UF} is not truthfully implementable by a direct mechanism.

We will now show that the allocation choice function $f^*$ is implementable in Bayes-Nash equilibrium.
{
\allowdisplaybreaks
Consider the mechanism $\cM_0 = ((\Sig_i)_{i}, h_0)$ with the signal sets for 
the players being 
$\Sig_1 = \Sig_2 = \{\MF^a, \MF^b, \UF, \SF \}$,
and the allocation function $h_0$ given by
\begin{align*}
	&h_0(\SF, \sig_2) = h_0(\sig_1, \SF) = c, \quad \forall \sig_1 \in \Sig_1, \sig_2 \in \Sig_2,\\
	&h_0(\UF, \UF) =  \{(1/2, a); (1/2, b)\},\\
	&h_0(\UF, \MF^a) = a,\\
	&h_0(\UF, \MF^b) = b,\\
	&h_0(\MF^a, \UF) = a,\\
	&h_0(\MF^b, \UF) = b,\\
	&h_0(\MF^a, \MF^a) = a,\\
	&h_0(\MF^b, \MF^b) = b,\\
	&h_0(\MF^a, \MF^b) = h_0(\MF^b, \MF^a)  = \{(1/2, a); (1/2, b)\}.
\end{align*}
}


Now consider the strategies $\sigma_1^*$ and $\sigma_2^*$ given by
\begin{align}   \label{eq:example-strategy}
	&\sigma_i^*(\SF) = \SF, \nonumber\\
	&\sigma_i^*(\UF) = \UF, \nonumber \\
	&\sigma_i^*(\MF) = \{(1/2, \MF^a); (1/2, \MF^b)\},
\end{align}
for $i = 1, 2$.

One can check that the allocation function $h_0$ and the strategy profile $\sigma^*$ induce the 
allocation choice function 
$f^*$ defined above. We will now verify that $\sigma^*$ is a Bayes-Nash equilibrium and 
thus 
conclude that $f^*$ is implementable in Bayes-Nash equilibrium.

If player $i$ has type $\SF$ then clearly $\SF$ is a best response signal for her. 
To see this, observe that amongst all the lotteries $L_i \in \Delta(\Gamma_i)$, $V_i^{\SF}(L_i)$ is maximized when $L_i = \III$ (this follows from the first order stochastic dominance property of CPT preferences).
Since signaling $\SF$ produces the lottery $\III$ for player $i$, we get that it is her best response.
If player $i$ has type $\UF$, then 
signaling 
$\UF$ dominates 
signaling 
$\SF$.
To see this, note that amongst all the lotteries $L_i \in \Delta(\Gamma_i)$, $V_i^{\UF}(L_i)$ is minimized when $L_i = \III$ (this follows from the first order stochastic dominance property of CPT preferences).
Since signaling $\SF$ produces the lottery $\III$ for player $i$, we get that it is dominated by all other strategies, in particular, signaling $\UF$.
As for comparing with signaling $\MF^a$ or 
$\MF^b$, if she signals $\UF$ 
then she will face the 
lottery 
\[
	L_i(\mu_i(\UF, \UF; \cM_0, F, \sigma^*_{-i})) = \{(7/32, \I); (7/32, \II); (1/8, \III); (7/32, \IV); (7/32, \V)\}.
\]
If she signals $\MF^a$, then she will face the lottery
\[
	L_i(\mu_i(\UF, \MF^a;\cM_0, F, \sigma^*_{-i})) = \{(3/8, \I); (1/16, \II); (1/8, \III); (1/16, \IV); (3/8, \V)\}.
\]
If she signals $\MF^b$, then she will face the lottery
\[
	L_i(\mu_i(\UF, \MF^b;\cM_0, F, \sigma^*_{-i})) = \{(1/16, \I); (3/8, \II); (1/8, \III); (3/8, \IV); (1/16, \V)\}.
\]
The CPT values in each of these scenarios are as follows:
\begin{align*}
	&V_i^{\UF}(L_i(\mu_i(\UF, \UF;\cM_0, F, \sigma^*_{-i}))) \\
	&= -50w^-(25/32) - 50w^-(18/32) - 90w^-(11/32) - 810w^-(4/32) \\
	&= -227.0312,
\end{align*}
\begin{align*}
	&V_i^{\UF}(L_i(\mu_i(\UF, \MF^a;\cM_0, F, \sigma^*_{-i}))) \\
	&= -50w^-(20/32) - 50w^-(18/32) - 90w^-(16/32) - 810w^-(4/32) \\
	&= -227.8125,
\end{align*}
and,
\begin{align*}
	&V_i^{\UF}(L_i(\mu_i(\UF, \MF^b;\cM_0, F, \sigma^*_{-i})))\\
	&= -50w^-(30/32) - 50w^-(18/32) - 90w^-(6/32) - 810w^-(4/32) \\
	& = -235.6250.
\end{align*}
Thus, signaling $\UF$ is the best response of a player with type $\UF$.

Finally, let player $i$ have type $\MF$. Depending on what she signals, we have the following lotteries:
\begin{align*}
	&L_i(\mu_i(\MF, \MF^a; \cM_0, F, \sigma^*_{-i})) = \{(3/8, \I); (1/16, \II); (1/8, \III); (1/16, \IV); (3/8, \V)\},\\
	&L_i(\mu_i(\MF, \MF^b; \cM_0, F, \sigma^*_{-i})) = \{(1/16, \I); (3/8, \II); (1/8, \III); (3/8, \IV); (1/16, \V)\},\\
	&L_i(\mu_i(\MF, \UF; \cM_0, F, \sigma^*_{-i})) = \{(7/32, \I); (7/32, \II); (1/8, \III); (7/32, \IV); (7/32, \V)\},\\
	&L_i(\mu_i(\MF, \SF; \cM_0, F, \sigma^*_{-i})) = \III.
\end{align*}
The corresponding CPT values are as follows:
\begin{align*}
	&V_i^{\MF}(L_i(\mu_i(\MF, \MF^a; \cM_0, F, \sigma^*_{-i})))\\
	&= 3.8w^+(20/32) + 2.016 w^+(18/32) + 2.8 w^+(14/32) + 5 w^+(12/32)\\
	&= 5.8243,
\end{align*}
\begin{align*}
	&V_i^{\MF}(L_i(\mu_i(\MF, \MF^b; \cM_0, F, \sigma^*_{-i}))) \\
	& = 3.8w^+(30/32) + 2.016 w^+(18/32) + 2.8 w^+(14/32) + 5 w^+(2/32)\\
	&= 5.8243,
\end{align*}
\begin{align*}
	&V_i^{\MF}(L_i(\mu_i(\MF, \UF; \cM_0, F, \sigma^*_{-i}))) \\
	&= 3.8w^+(25/32) + 2.016 w^+(18/32) + 2.8 w^+(14/32) + 5 w^+(7/32)\\
	&= 5.6993,
\end{align*}
and,
\begin{align*}
	&V_i^{\MF}(L_i(\mu_i(\MF, \SF; \cM_0, F, \sigma^*_{-i}))) = 5.816.
\end{align*}
Thus $\sigma_i^*(\MF)$ has support on optimal signals, and hence is a best response. This completes the verification that $\sigma^*$ is a Bayes-Nash equilibrium.
With this, we end our example.
\qed
\end{example}



In the previous example, 
let us focus on the behavior of player $i$ when she has type $\MF$.
For any mechanism 
with the signal sets for 
the players being 
$\Sig_1 = \Sig_2 = \{\MF^a, \MF^b, \UF, \SF \}$ as above 
(the mechanism 
$\cM_0 = ((\Sig_i)_{i}, h_0)$ 
considered above is an instance of such a mechanism), 
the
signals $\MF^a$ and $\MF^b$ allow this player to play so that the lotteries faced by her are 
$L_i' := L_i(\mu_i(\MF, \MF^a); \cM, F, \sigma_{-i})$ and 
$L_i'' := L_i(\mu_i(\MF, \MF^b); \cM, F, \sigma_{-i})$ respectively,
where $F$ denotes the prior distribution on types 
(i.e. the product distribution with marginals given as in 
\eqref{eq: ex_prior1} above) and $\sigma_{-i}$ denotes the strategy of the other player.
The lotteries $L_i'$ and $L_i''$ are equally preferred by 
\black
player $i$ when she has type $\MF$,
and they are preferred over the lotteries corresponding to signaling $\UF$ or $\SF$,
when the mechanism is $\cM_0 = ((\Sig_i)_{i}, h_0)$ as considered in example~\ref{ex: no_rev_impl}, and the other player plays
according to the strategy prescribed in \eqref{eq:example-strategy}.
Under the equilibrium strategy 
$\sigma_i^*$, as defined in \eqref{eq:example-strategy}, 
when player $i$ has type $\MF$ she signals $\MF^a$ or $\MF^b$ each with probability half.

We can think of player $1$ as tossing a fair coin to decide whether to signal $\MF^a$ or $\MF^b$ when
her type is $\MF$, and similarly for player $2$.
The outcome of the coin toss is private knowledge to the player tossing the coin.
The equilibrium strategies in \eqref{eq:example-strategy} correspond to each player signaling 
$\UF$ or $\SF$ truthfully if that is her type, while if
her type is $\MF$ then 
she signals $\MF^a$ or $\MF^b$ depending on the outcome of her coin toss.
From our analysis in the above example, we have that 
these strategies form an
$F$-Bayes-Nash equilibrium 
for this game and 
induce 
the allocation choice function $f^*$.

An alternate viewpoint is to think of 
the coins being tossed
at the beginning as before,
but now let us assume that the system operator observes the outcomes of both the coins.
We continue to assume that 
each player does not observe the result of the coin toss of the other player.
Suppose each player 
only has the option 
to signal from $\{\MF, \UF, \SF\}$.
The system operator collects these signals and implements a lottery on the allocation set according to the following rule: 
If player $i$ signals $\UF$ or $\SF$ then let $\sig_i' = \UF$ or $\sig_i' = \SF$ respectively. 
If player $i$ signals $\MF$ then, depending on the outcome of coin toss $i$, let $\sig_i' = \MF^a$ or $\MF^b$.
The system operator implements $h_0(\sig_1', \sig_2')$.
Now consider the strategy where each player reports her type truthfully.
We observe that this strategy is 
an $F$-Bayes-Nash equilibrium 
for this game and induces $f^*$.


Thus the issue with the revelation principle is superficial in the sense that the reason that it does not hold is not that player $i$ does not wish to reveal her type, but 
rather that she 
would like to have 
an
asymmetry in the information of the players.
In the above example, this asymmetry 
comes from the coin tosses and, as seen in the latter viewpoint, these coin tosses can be thought of as 
shared between each individual player and the system operator, so one could even think of the coins as being
tossed by system operator, with the result of each individual coin toss being shared with the respective player.
To capture this intuition, we propose a framework where there is a mediator who sends messages to each individual player before collecting their signals.
As we will see now, 
this way we can recover a form of the revelation principle.


\subsection{Mediated mechanisms and the revelation principle}
\label{subsec: med_mech_frame}

We now lay out the framework for a mechanism with messages from the mediator, along the lines of the 
augmented
framework 
for mechanism design 
motivated by the example above.
Let $\Msg_i$ be a finite \emph{message set} for each player $i$, with a typical element denoted by $\msg_i$, and
let $\Msg := \prod_i \Msg_i$.
Let $\mdt \in \Delta(\Msg)$ denote a \emph{mediator distribution} from which the mediator draws a profile of messages $\msg := (\msg_1, \dots, \msg_n)$.
Let $\mdt_i \in \Delta(\Msg_i)$ denote the marginal of $\mdt$ on $\Msg_i$.
For any $\msg_i \in \supp \mdt_i$,
 let the conditional distribution be given by
\begin{equation}
\label{eq: E_cond}
	\mdt_{-i}(\msg_{-i}|\msg_i) := \frac{\mdt(\msg_i, \msg_{-i})}{\mdt_i(\msg_i)}, 
	\text{ for all } \msg_{-i} \in \Msg_{-i},
\end{equation}
where $\msg_{-i} := (\msg_j)_{j \neq i}$ and $\Msg_{-i} := \prod_{j \neq i} \Msg_j$.
Let $\Sig_i$ be a finite set of signals as before.
Let 
\begin{equation}
	\label{eq: def_med_allocation_func}
	h: \Msg \times \Sig \to \Delta(A)
\end{equation}
be a \emph{mediated allocation function}.
The message sets $\Msg_i, i \in [n]$, a mediator distribution $\mdt \in \Delta(\Msg)$, and a mediated allocation function $h$ together constitute a \emph{mediated mechanism}, denoted by 
\begin{equation}
\label{eq: def_med_mech}
	\cal{M} := ((\Msg_i)_{i \in [n]}, \mdt, (\Sig_i)_{i \in [n]}, h).
\end{equation}

The mediator first draws a profile of messages $\msg$ from the distribution $\mdt$.
Each player $i$ observes her message $\msg_i$, and then sends a signal $\sig_i$ to the mediator.
The mediator collects the signals from all the players and 
then chooses an allocation 
according to 
the probability distribution 
$h(\msg, \sig)$.
A strategy for any player $i$ is thus given by 
\begin{equation}\label{eq: def_tau}
	\tau_i: \Msg_i \times \Type_i \to \Delta(\Sig_i).
\end{equation}
Let $\tau_i(\sig_i|\msg_i, \type_i)$ denote the probability of signal $\sig_i$ under the distribution $\tau_i(\msg_i, \type_i)$.
Let $\tau := (\tau_1, \dots, \tau_n)$ denote the profile of strategies.
Suppose player $i$ has received message $\msg_i$ and has type $\type_i$ (thus, $\msg_i \in \supp \mdt_i$, and $\type_i \in \supp F_i$), and she chooses to signal $\sig_i$
(so $\sig_i \in \supp \tau_i(\msg_i,\type_i)$); 
then consider the probability distribution 
$\mu_i(\msg_i, \type_i, \sig_i ; \cM, F, \tau_{-i}) \in \Delta(A)$ 
given by
\begin{align}
\label{eq: mu_i_med}
	\mu_i(\alpha | \msg_i, \type_i, \sig_i ; \cM, F, \tau_{-i}) 
	:= \sum_{\msg_{-i}} \mdt_{-i}(\msg_{-i} | \msg_i)& \sum_{\type_{-i}} F_{-i}(\type_{-i}|\type_i) \nn \\
	& \times \sum_{\sig_{-i}} \prod_{j \neq i} \tau_j(\sig_j | \msg_j, \type_j) h(\alpha | \msg, \sig).
\end{align}

The \emph{best response strategy set} $BR_{i}(\tau_{-i})$ of player $i$ to a strategy profile $\tau_{-i}$ of 
her 
opponents consists of all strategies $\tau_i^*: \Msg_i \times \Type_i \to \Delta(\Sig_i)$ such that
\begin{equation}
W_i^{\type_i}(\mu_i(\msg_i, \type_i, \sig_i ; \cM, F, \tau_{-i})) \geq W_i^{\type_i}(\mu_i(\msg_i, \type_i, \sig_i' ; \cM, F, \tau_{-i})),
\end{equation}
for all $\msg_i \in \supp \mdt_i, \type_i \in \supp F_i, \sig_i \in \supp \tau_i^*(\msg_i, \type_i), \sig_i' \in \Sig_i$.

	A strategy profile $\tau^*$ is said to be an \emph{$F$-Bayes-Nash equilibrium} for the environment $\cal{E}$ with respect to the mediated mechanism $\cM$ if for each player $i$ we have
	\begin{equation}
	\label{eq: def_med_Bayes_Nash_eq}
		\tau_i^* \in BR_i(\tau^*_{-i}).
	\end{equation}

We will say that an allocation choice function $f: \Type \to \Delta(A)$ is implementable in $F$-Bayes-Nash equilibrium by a mediated mechanism if there exists a mediated mechanism $\cal{M}$ and an $F$-Bayes-Nash equilibrium $\tau$ with respect to this mediated mechanism such that $f$ is 
the induced allocation choice function, 
i.e. 
for all $\type \in \supp F, \alpha \in A$, 
we have
\begin{equation}
\label{eq: f_induced}
	f(\alpha | \type) = \sum_{\msg} \mdt(\msg) \sum_{\sig} \l(\prod_{i} \tau_i(\sig_i|\msg_i, \type_i)\r) h(\alpha | \msg, \sig).
\end{equation}

A mediated mechanism $\cal{M} = ((\Msg_i)_{i \in [n]}, \mdt, (\Sig_i)_{i \in [n]}, h)$ is called a \emph{direct mediated mechanism} if $\Sig_i = \Type_i$ for all $i$,
and we write it as 
$\cal{M}^d = ((\Msg_i)_{i \in [n]}, \mdt, (\Type_i)_{i \in [n]}, h^d)$,
where 
$$h^d: \Msg \times \Type \to \Delta(A)$$ 
is the corresponding
\emph{direct mediated allocation function}.

For a direct mediated mechanism, the truthful strategy $\tau_i^d$ for player $i$ should satisfy
$\tau_i^d(\msg_i, \type_i) = \type_i,$
for all $\msg_i \in \Msg_i,$ 
and $\type_i \in \Type_i$. 
Thus, if player $i$ receives a message $\msg_i$ and has type $\type_i$, she
reports her true type 
$\type_i$ 
irrespective of her received message. 
In a way, the messages are present only to align the beliefs of the players appropriately so that truth-telling is an equilibrium strategy (depending on the type of equilibrium 
under consideration, i.e.
Bayes-Nash, dominant, or belief-dominant
equilibrium).
Note that in the definition of the truthful strategy $\tau_i^d$ for player $i$ 
we require
$\tau_i^d(\msg_i, \type_i) = \type_i,$
for all $\type_i \in \Type_i$ and $\msg_i \in \Msg_i,$ 
and not just for $\type_i \in \supp F_i$ (when discussing an $F$-Bayes-Nash equilibrium)
and $\msg_i \in \supp D_i$. This is done to make the notion of a truthful strategy uniquely 
defined.

An allocation choice function 
$f$
is said to be truthfully implementable in mediated 
$F$-Bayes-Nash equilibrium 
if there exists a direct mediated mechanism $\cM^d$ such that the truthful strategy profile $\tau^d$ is a mediated 
$F$-Bayes-Nash equilibrium 
and it implements $f$.

Let
\begin{equation}
\label{eq: induce_lot_med_dom_eq}
	\mu_i(\msg_i, \type_i, \sig_i; \cM, \sig_{-i}) := \sum_{\msg_{-i}} \mdt_{-i}(\msg_{-i} | \msg_i) h(\msg, \sig),
\end{equation}
denote the lottery faced by player $i$ with type $\type_i$, who has received message $\msg_i$ (thus, $\msg_i \in \supp \mdt_i$) and believes that her opponents are going to report $\sig_{-i}$.
Similarly, let
\begin{equation}
\label{eq: induce_lot_med_belief_dom_eq}
	\mu_i(\msg_i, \type_i, \sig_i; \cM, G_{-i}) := \sum_{\msg_{-i}} \mdt_{-i}(\msg_{-i} | \msg_i) \sum_{\sig_{-i}} G_{-i}(\sig_{-i}) h(\msg, \sig),
\end{equation}
denote the lottery faced by player $i$ with type $\type_i$, who has received message 
$\msg_i \in \supp D_i$
and has belief $G_{-i} \in \Delta(\Sig_{-i})$ over her opponents' signal profiles.
We define strategy 
$\tau^*_i$ 
to be dominant if, for all $\msg_i \in \supp \mdt_i$, 
$\type_i \in \Type_i$, 
$\sig_i \in \supp \tau_i^*(\msg_i, \type_i)$, 
$\sig_i' \in \Sig_i$, 
and $\sig_{-i} \in \Sig_{-i}$, we have
\begin{equation}
\label{eq: med_dominant_strat}
	W_i^{\type_i}(\mu_i(\msg_i, \type_i, \sig_i; \cM, \sig_{-i})) \geq W_i^{\type_i}(\mu_i(\msg_i, \type_i, \sig_i'; \cM, \sig_{-i})).
\end{equation}
Similarly,
we define strategy 
$\tau^*_i$ 
to be belief-dominant if, for all $\msg_i \in \supp \mdt_i$, 
$\type_i \in \Type_i$, 
$\sig_i \in \supp \tau_i^*(\msg_i, \type_i)$, 
$\sig_i' \in \Sig_i$, 
and $G_{-i} \in \Delta(\Sig_{-i})$, we have
\begin{equation}
\label{eq: med_belief_dominant_strat}
	W_i^{\type_i}(\mu_i(\msg_i, \type_i, \sig_i; \cM, G_{-i})) \geq W_i^{\type_i}(\mu_i(\msg_i, \type_i, \sig_i'; \cM, G_{-i})).
\end{equation}

An allocation choice function $f$ is said to be \emph{implementable in dominant equilibrium}
by a mediated mechanism if there is a mediated mechanism $\cM$ and a dominant equilibrium $\tau$
(i.e. a strategy profile comprised of dominant strategies for the individual players) such that
$f$ is the allocation choice function induced by $\tau$ under $\cM$, i.e. \eqref{eq: f_induced} holds for
all $\type \in \Type$ and $\alpha \in A$.
$f$ is said to be \emph{truthfully implementable in dominant equilibrium} by a direct mediated
mechanism if there is a directed mediated mechanism $\cM^d$ such that the truthful strategy profile
is a dominant equilibrium and
induces $f$ under $\cM^d$. The notions of implementability by a mediated mechanism
and truthful implementability by a direct mediated mechanism of
an allocation choice function in belief-dominant equilibrium can be similarly defined.

If the message set $\Msg_i$  is a singleton for each player $i$, then we get back the original mechanism design framework.
Thus, the mediated mechanism design framework defined above is a generalization of the mechanism design framework.
This generalization allows us to establish a form of the revelation principle even when players have CPT preferences.

A special case of the mediated mechanism design framework is when the mediator message profile $\msg$ is publicly known.
That is, each player knows the entire message profile instead of privately knowing only her own message.
This would happen if $\Msg_i = \Msg_*$, for all $i \in [n]$, and 
$\mdt$ is a diagonal distribution, i.e. 
$\mdt(\msg) = 0$ for all message profiles $\msg = (\msg_i)_{i \in [n]}$ such that $\msg_i \neq \msg_j$ for some pair $i, j \in [n]$.
Let $\Msg_*$ denote the common message set and let $\mdt_* \in \Delta(\Msg_*)$ denote the mediator distribution on this set.
Let 
$$\cM_* := (\Msg_*, \mdt_*, (\Sig_i)_{i \in [n]}, h_*)$$
denote such a mediated mechanism with common messages,
where now
\begin{equation*}
	h_*: \Msg_* \times \Sig \to \Delta(A).
\end{equation*}
We will call 
$\cM_*$
a \emph{publicly mediated mechanism}.
The notions
of an allocation choice function being 
implementable 
in publicly mediated Bayes-Nash equilibrium, publicly mediated dominant 
equilibrium,
or 
publicly mediated belief-dominant 
equilibrium
can be defined similarly
to the corresponding earlier definitions that were made for general message sets. 
The notions
of an allocation choice function being 
truthfully implementable 
in direct publicly mediated Bayes-Nash equilibrium, direct publicly mediated dominant equilibrium, 
or 
direct publicly mediated belief-dominant equilibrium can also be defined similarly
to the corresponding earlier definitions that were made for general message sets. 

We are now in a position to state one of our main results.

\begin{theorem}[Revelation Principle]
\label{thm: rev_prin_1}
We have the following three versions of the revelation principle:
\begin{enumerate}[(i)]
	\item If an allocation choice function is implementable in Bayes-Nash equilibrium by a mediated mechanism, then it is also truthfully implementable in Bayes-Nash equilibrium by a 
	direct mediated 
	mechanism.
	\item If an allocation choice function is implementable in dominant equilibrium by a publicly mediated mechanism, then it is also truthfully implementable in dominant equilibrium by a 
	direct publicly mediated
	mechanism.
	\item If an allocation choice function is implementable in belief-dominant equilibrium by a mediated (resp. publicly mediated) mechanism, then it is also truthfully implementable in belief-dominant equilibrium by a 
	direct
	mediated (resp. 
	direct
	publicly mediated) 
	mechanism.
\end{enumerate}
\end{theorem}

We prove this theorem in Appendix~\ref{sec: rev_proof}.
Theorem~\ref{thm: rev_prin_1}, in particular, implies that if an allocation choice function is implementable in Bayes-Nash equilibrium by a non-mediated mechanism then it is truthfully implementable in Bayes-Nash equilibrium by a direct mediated mechanism.
Similarly, if an allocation choice function is implementable in dominant strategies (resp. belief-dominant strategies) by a non-mediated mechanism, then it is truthfully implementable in dominant strategies (resp. belief-dominant strategies) by a 
direct publicly mediated
mechanism.
Table~\ref{tab: rev_prin} summarizes the different implementability settings under which the revelation principle does and does not hold.
Example~\ref{ex: no_rev_impl} shows that the revelation principle does not hold for the setting with Bayes-Nash equilibrium and non-mediated mechanism.
In example~\ref{ex: no_rev_impl_triv}, we show that the revelation principle does not hold for the settings with dominant equilibrium or belief-dominant equilibrium and non-mediated mechanism.
In example~\ref{ex: Bayes_Nash_no_rev_public}, we show that the revelation principle does not hold for the settings with Bayes-Nash equilibrium and publicly mediated mechanism.
The question of whether the revelation principle holds in the setting with dominant equilibrium and mediated mechanism
remains 
open for future investigation.

\begin{table}
\centering
\begin{tabular}{| c | c | c | c |}
\hline
 & Non-mediated & Publicly Mediated & Mediated \\
\hline
Bayes-Nash Equilibrium & \xmark & \xmark & \cmark\\ 
Dominant Strategies & \xmark & \cmark & \textbf{?} \\
Belief-dominant Strategies & \xmark & \cmark & \cmark \\
\hline
\end{tabular}
\caption{Settings in which the revelation principle holds.}
\label{tab: rev_prin}
\end{table}

\begin{example}
\label{ex: no_rev_impl_triv}
Consider the setting from example~\ref{ex: not_belief_dominant_VCG} with two players.
Recall that
$\Type_1 = \Type_2 = \{\Up, \D\}$, $A = \{a, b, c\}$, $\Gamma_1 = \{\I, \II, \III, \IV, \V\}$, $\Gamma_2 = A$.
The allocation-outcome mapping is given by the product distribution of the marginals $\zeta_1$ and $\zeta_2$, given by
$\zeta_1(a) = \{(1/2, \I); (1/2, \V) \}, \zeta_1(b) = \{(1/2, \II); (1/2, \IV) \}, \zeta_1(c) = \{(1, (\III)) \},$
and 
	$\zeta_2(\alpha) = \alpha, \forall \alpha \in A$.
The probability weighting functions for gains for the two players are
\[
	w_1^+(p) = \exp \{-(-\ln p)^{0.5}\}, w_2^+(p) = p,
\]
for $p \in [0,1]$
(see figure~\ref{fig: wt_pwf_prelec}).
Let the value functions $v_1$ and $v_2$ be given by
\begin{center}
\begin{tabular}{| c | c | c | c | c | c |}
\hline
$v_1$ & I & II & III & IV & V\\
\hline
$\Up$ & $2x$ & $x+1$ & $1.99$ & $1$ & $0$\\ 
\hline
$\D$ & $0$ & $0$ & $1$ & $0$ & $0$\\ 
\hline
\end{tabular}
\hspace{1in}
\begin{tabular}{| c | c | c | c | c | c |}
\hline
$v_2$ & a & b & c\\
\hline
$\Up$ & $1$ & $0$ & $2$\\ 
\hline
$\D$ & $0$ & $1$ & $2$\\ 
\hline
\end{tabular}
\end{center}
where $x := 1/w_1^+(0.5) = 2.2992$.

Consider a mechanism 
$\cM_0 = \{ (\Sig_1, \Sig_2), h_0\}$, 
where
$\Sig_1 = \{a,b,c\}$, $\Sig_2 = \{\Up, \D\}$, and
\begin{align*}
	h_0(a, \sig_2) &= a, \\
	h_0(b, \sig_2) &= b, \\
	h_0(c, \sig_2) &= c,
\end{align*}
for all $\sig_2 \in \Sig_2$.
The CPT values 
for player $1$ having type $\Up$ for the lotteries over her outcomes corresponding to the different allocations 
are
given by
\begin{align*}
	V_1^\Up(L_1(a)) &= 2xw_1^+(0.5) = 2,\\
	V_1^\Up(L_1(b)) &= w_1^+(1) + x w_1^+(0.5) = 2,\\
	V_1^\Up(L_1(c)) &= 1.99.
\end{align*}
Further, the CPT values
for player $1$ having type $\D$ for the above lotteries 
are
given by
\begin{align*}
	V_1^\D(L_1(a)) &= 0, &V_1^\D(L_1(b)) &= 0, &V_1^\D(L_1(c)) &= 1.
\end{align*}
Since the allocation choice function $h_0$ does not depend on the signal of player $2$, from the above the above calculations,
we observe that the strategy $\sigma_1$ given by
\begin{align*}
\sigma_1(\cdot|0) &= \{(0.5, a); (0.5, b)\},\\
\sigma_1(\cdot|1) &= c,
\end{align*}
is a dominant strategy.
and a belief-dominant strategy.
Let $\sigma_2$ be the truthful strategy for player $2$.
Again, since the allocation choice function $h_0$ does not depend on the signal of player $2$, $\sigma_2$ is trivially a dominant strategy and a belief-dominant strategy.
Thus $\sigma = (\sigma_1, \sigma_2)$ is a dominant equilibrium and a belief-dominant equilibrium.
The corresponding social choice function $f$ is given by
\begin{align*}
f(\Up, \type_2) &= \{(0.5, a);(0.5, b)\}, \\
f(\D, \type_2) &= c.
\end{align*}
Thus, the allocation choice function $f$ is implementable in dominant (resp. belief-dominant) equilibrium. 
Suppose there 
were 
a direct mechanism $\cM_0^d = h_0^d$ that truthfully implements the allocation choice function $f$ in dominant (resp. belief-dominant) equilibrium.
Then, $h_0^d = f$.
As observed in example~\ref{ex: not_belief_dominant_VCG}, the CPT value for player $1$ having type $\Up$ for the lottery corresponding to $\{(0.5, a); (0.5; b) \}$ is
\[
	V_1^{\Up}(L_1(\{(0.5, a); (0.5; b) \})) = 1.9851.
\] 
If player $1$ has type $\Up$ and believes that player $2$'s type report is $\Up$ (or equivalently, any other distribution over player $2$'s type report),
then player $1$ would deviate from her truthful strategy and report $\D$ instead, because it gives her a higher CPT value.
Hence the truthful strategy $\sigma_1^d$ is not 
a 
dominant (resp.  belief-dominant) equilibrium for the direct mechanism $\cM_0^d$.
Thus $f$ is not truthfully implementable in dominant (resp.  belief-dominant) equilibrium by a direct mechanism.
\qed
\end{example}


We will now show that the revelation principle does not hold for the setting with Bayes-Nash equilibrium and publicly mediated mechanism.
Let us first make an observation regarding the allocation choice functions that are truthfully implementable in $F$-Bayes-Nash equilibrium by a 
direct publicly mediated
mechanism.
Let $f$ be an allocation choice function that is truthfully implementable in $F$-Bayes-Nash equilibrium by a 
direct publicly mediated
mechanism 
\begin{equation}
	\cM_*^d = (\Msg_*, \mdt_*, (\Theta_i)_{i \in [n]}, h_*^d),
\end{equation}
where
\begin{equation}
	h_*^d : \Msg_* \times \Type \to \Delta(A),
\end{equation}
is the direct mediated allocation function for this 
direct publicly mediated
mechanism.
Since truthful strategies $\tau^d$ are an $F$-Bayes-Nash equilibrium, for each $\msg_* \in \supp \mdt_*$, we have
\begin{align}
\label{eq: direct_pub_mech_BNE_ineq}
	W_i^{\type_i}(\mu_i(\msg_*, \type_i, \type_i; \cM_*^d, F, \tau_{-i}^d)) \geq W_i^{\type_i}(\mu_i(\msg_*, \type_i, \tilde \type_i; \cM_*^d, F, \tau_{-i}^d)),
\end{align}
for all $\type_i \in \supp F_i, \tilde \type_i \in \Type_i, i \in [n]$,
where
\begin{align*}
	\mu_i(\msg_*, \type_i, \tilde \type_i; \cM_*^d, F, \tau_{-i}^d) 
	= \sum_{\type_{-i}} F_{-i}(\type_{-i} | \type_i) h_*^d(\msg_*, \tilde \type_i, \type_{-i}),
\end{align*}
is the lottery induced on the allocations for player $i$ receiving message $\msg_*$, having type $\type_i$, and deciding to report type $\tilde \type_i$.
Now,
fix $\msg_* \in \Msg_*$ with $\mdt_*(\msg_*) > 0$,
and consider a non-mediated direct mechanism 
$\cM^d_0 := ((\Type_i)_{i \in [n]}, h_0^d)$,
with its direct allocation function 
being $h_0^d(\cdot) := h_*^d(\msg_*, \cdot) : \Type \to \Delta(A)$.
It follows from \eqref{eq: direct_pub_mech_BNE_ineq} that truthful strategies corresponding to mechanism $\cM_0^d$ 
form an 
$F$-Bayes-Nash equilibrium.
Thus, we note that $h_*^d(\msg_*, \cdot)$ is the allocation function truthfully implemented by the non-mediated direct mechanism $\cM_0^d$.
Since mechanism $\cM_*^d$ truthfully implements the allocation function $f$ in $F$-Bayes-Nash equilibrium, we have that
\begin{align*}
	f(\type) = \sum_{\msg_*} \mdt_*(\msg_*) h_*^d(\msg_*, \type),
\end{align*}
for all $\type \in \supp F$.
From these two observations, we conclude that if $f$ is an allocation choice function that is truthfully implementable in $F$-Bayes-Nash equilibrium by a 
direct publicly mediated
mechanism, then $f$ is a convex combination of allocation choice functions each of which 
is
truthfully implementable in $F$-Bayes-Nash equilibrium by a non-mediated direct mechanism.
It is easy to see that the converse of this statement is also true.

In the following example, we will use this observation to establish that the revelation principle does not hold for the setting with Bayes-Nash equilibrium and publicly mediated mechanism.

\begin{example}
	\label{ex: Bayes_Nash_no_rev_public}
	Let there be two players, i.e. $n = 2$.
	Let $\Type_1 = \Type_2 = \{\Up, \D \}$.
	Let $\Gamma_1 = \Gamma_2 = \{\I, \II, \III, \IV, \V \}$.
	Let the value function $v_1$ for player $1$ be as shown below
	\begin{center}
	\begin{tabular}{| c | c | c | c | c | c |}
	\hline
	$v_1$ & I & II & III & IV & V\\
	\hline
	$\Up$ & $80$ & $57$ & $34$ & $17$ & $0$\\ 
	\hline
	$\D$ & $0$ & $0$ & $100$ & $0$ & $0$\\ 
	\hline
	\end{tabular}
	\end{center}
	and let the value function $v_2$ for player $2$ be as shown below
	\begin{center}
	\begin{tabular}{| c | c | c | c | c | c |}
	\hline
	$v_2$ & I & II & III & IV & V\\
	\hline
	$\Up$ & $-79$ & $-56$ & $-33$ & $-17$ & $0$\\ 
	\hline
	$\D$ & $0$ & $0$ & $100$ & $0$ & $0$\\ 
	\hline
	\end{tabular}
	\end{center}
	Let the probability weighting functions for both the players, for both types, for gains and losses, be given by the following piecewise linear function:
	\begin{equation*}
	w_1^{\pm}(p) = w_2^{\pm}(p) = w(p) = \begin{cases}
		(8/7)p, &\text{ for } 0 \leq p < (7/32),\\
		(1/4) + (2/3)(p-7/32), &\text{ for } (7/32) \leq p < 25/32,\\
		(5/8) + (12/7)(p - 25/32), &\text{ for } (25/32) \leq p < 1,
	\end{cases}
	\end{equation*}
	(See the probability weighting function for gains in figure~\ref{fig: ex_weight}.)
	Let the prior distribution $F$ be such that the types of the players are independently sampled with probabilities, 
\begin{equation}\label{eq: ex_prior}
	\bbP(\Up) = 3/4, \bbP(\D) = 1/4.	
\end{equation}
	Let $A = \{a, b, c\}$.
	Let
	\begin{align*}
	\zeta(a) &= \{(1/2, (\I,\I)); (1/2, (\V, \V)) \},\\
	\zeta(b) &= \{(1/2, (\II,\II)); (1/2, (\IV, \IV)) \},\\
	\zeta(c) &= (\III,\III).
\end{align*}
Consider the allocation choice function $f^*$ given by
\begin{align*}
	&f^*(\D, \type_2) = f^*(\type_1, \D) = c, \quad \forall \type_1 \in \Type_1, \type_2 \in \Type_2\\
	&f^*(\Up, \Up) = \{(1/2, a); (1/2, b)\}.
\end{align*}

We will now show that $f^*$ is implementable in $F$-Bayes-Nash equilibrium by a publicly mediated mechanism.
In fact, we will show that $f^*$ is implementable in $F$-Bayes-Nash equilibrium by a 
non-mediated mechanism.
We will then show that $f^*$ cannot be a convex combination of allocation choice functions each of which 
is
truthfully implementable by a 
non-mediated direct mechanism.
This will give us that $f^*$ is not truthfully implementable in $F$-Bayes-Nash equilibrium by a 
direct publicly mediated
mechanism.
We will then conclude that the revelation principle does not hold for the setting with Bayes-Nash equilibrium and publicly mediated mechanism.

Consider the mechanism $\cM_0 = ((\Sig_i)_{i \in [n]}, h_0)$, where $\Sig_1 = \{\Up^a, \Up^b, \D\}$, $\Sig_2 = \{ \Up, \D\}$, and the allocation function $h_0$ is given by
\begin{align*}
	&h_0(\D, \sig_2) = h_0(\sig_1, \D) = c, \quad \forall \sig_1 \in \Sig_1, \sig_2 \in \Sig_2,\\
	&h_0(\Up^a, \Up) = a,\\
	&h_0(\Up^a, \Up) = b.
\end{align*}
Consider the strategy $\sigma_1$ for player $1$ given by
\begin{align*}
	&\sigma_1(\Up) = \{(1/2,\Up^a); (1/2, \Up^b)\},\\
	&\sigma_1(\D) = \D,
\end{align*}
and the strategy $\sigma_2$ for player $2$ given by
\begin{align*}
	&\sigma_2(\Up) = \Up,\\
	&\sigma_2(\D) = \D.
\end{align*}
It is easy to see that this induces the allocation choice function $f^*$.

We will now verify that $\sigma$ is an $F$-Bayes-Nash equilibrium for $\cM_0$.
If player $1$ has type $\Up$, then the CPT values of the lotteries faced by her corresponding to her signals are as follows:
\begin{align*}
	W_1^{\Up}(\mu_1(\Up, \Up^a; \cM_0, F, \sigma_{-1})) &= V_1^{\Up}(\{(3/8, \I); (0, \II); (1/4, \III); (0, \IV); (3/8, \V)\})\\
	&= 46w(3/8) + 34w(5/8)\\
	&= 34.
\end{align*}
\begin{align*}
	W_1^{\Up}(\mu_1(\Up, \Up^b; \cM_0, F, \sigma_{-1})) 
	&= V_1^{\Up}(\{(0, \I); (3/8, \II); (1/4, \III); (3/8, \IV); (0, \V)\})\\
	&= 23w(3/8) + 17w(5/8) + 17\\
	&= 34. 
\end{align*}
\begin{align*}
	W_1^{\Up}(\mu_1(\Up, \D; \cM_0, F, \sigma_{-1})) 
	&= V_1^{\Up}(\III)= 34. 
\end{align*}
Thus player $1$ is indifferent between all signals when she has type $\Up$ and so the 
strategy of signaling $\sigma_1(\Up) = \{(1/2,\Up^a); (1/2, \Up^b)\}$ is optimal for her.

If player $1$ has type $\D$, then $\III$ is the best outcome and she receives this lottery if she signals $\D$.
Thus $\D$ dominates any other strategy, in particular, signaling $\Up^a$ or $\Up^b$.

If player $2$ has type $\Up$, then the CPT values of the lotteries faced by her corresponding to her signals are as follows:
\begin{align*}
	W_2^{\Up}(\mu_1(\Up, \Up; \cM_0, F, \sigma_{-2}))  &= V_1^{\Up}(\{(3/16, \I); (3/16, \II); (1/4, \III); (3/16, \IV); (3/16, \V)\})\\
	&= - 23w(3/16) - 23w(3/8) - 16w(5/8) - 17w(13/16)\\
	&= -32.94.
\end{align*}
\begin{align*}
	W_2^{\Up}(\mu_1(\Up, \D; \cM_0, F, \sigma_{-2}))  &= V_1^{\Up}(\III)= -33. 
\end{align*}
Hence the strategy of signaling $\sigma_2(\Up) = \Up$ is optimal for player $2$ when she has type $\Up$.

If player $2$ has type $\D$, then $\III$ is the best outcome and she receives this lottery if she signals $\D$.
Thus $\D$ dominates any other strategy, in particular, signaling $\Up$.

This shows that $\sigma$ is an $F$-Bayes-Nash equilibrium for $\cM_0$, and hence establishes that $f^*$ is implementable in $F$-Bayes-Nash equilibrium by a 
non-mediated
mechanism.

Suppose $f^*$ 
were
a convex combination of allocation choice functions each of which 
is
truthfully implementable by a 
non-mediated direct mechanism.
Let $f$ be one of the allocation choice functions in this convex combination.
Since $f^*(\D, \type_2) = f^*(\type_1, \D) = c$ for all $\type_1, \type_2$, and since $\{c\}$ is an extreme point of the simplex $\Delta(A)$, we get that
\begin{equation}
\label{eq: pub_med_no_rev_BN_prop1}
	f(\D, \type_2) = f(\type_1, \D) = c, \quad  \forall  \type_1 \in \Type_1, \type_2 \in \Type_2.
\end{equation}
Similarly, since $f^*(\Up, \Up)$ lies on the line joining the vertices $\{a\}$ and $\{b\}$ of the simplex $\Delta(A)$, we get that
\begin{equation}
\label{eq: pub_med_no_rev_BN_prop2}
	f(\Up, \Up) = \{(x,a);(1-x,b)\},
\end{equation}
where $0 \le x \le 1$.

Let $f$ be truthfully implementable in $F$-Bayes-Nash equilibrium by the 
non-mediated
direct mechanism $\cM_0^d = h_0^d$.
Then $h_0^d = f$.
If player $1$ has type $\Up$, then the lottery faced by her if she reports $\Up$ is given by
\[
	L_1(\mu_1(\Up, \Up; \cM_0^d, F, \sigma_{-1}^d)) = 
	\{(3x/8, \I); (3(1-x)/8, \II); (1/4, \III); (3(1-x)/8, \IV); (3x/8, \V)\},
\]
where $\sigma_{-1}^d = \sigma_2^d$ is the truthful strategy of player $2$.
Let
\[
	E_3(x) := 
	23w\l(\frac{3x}{8}\r) + 23w\l(\frac{3}{8}\r) + 17w\l(\frac{5}{8}\r) + 17w\l(\frac{1-3x}{8}\r), 
\]
for $x \in [0,1]$.
We observe that $E_3(x)$ is maximum at $x = 0$ and $x = 1$, and for all $x \in (0,1)$, $E_3(x) < 34$. (See figure~\ref{fig: ex_valE3}.)

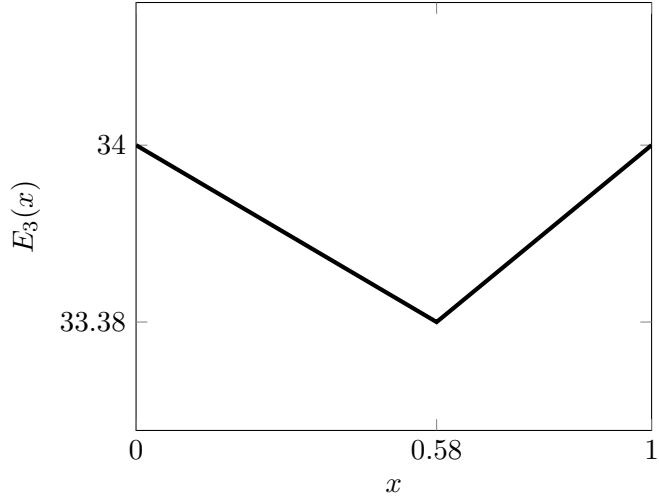
\begin{figure}
\centering
\begin{tikzpicture}[scale = 1]
    \begin{axis}[
    	xmin = 0,
    	xmax = 1,
    	ymin = 33,
    	ymax = 34.5,
        xlabel=$x$,
        ylabel=$E_3(x)$,
        xtick=\empty,
        ytick=\empty,
        extra x ticks={0,0.583,1},
        extra y ticks={34,33.38},
        every axis plot/.append style={ultra thick},
        ]
\addplot []
  table[row sep=crcr]{%
	0   34\\
    0.583   33.38\\
    1   34\\
};


\end{axis}
\end{tikzpicture}%
\caption{Plot of expression $E_3(x)$ in example~\ref{ex: Bayes_Nash_no_rev_public}.}
\label{fig: ex_valE3}
\end{figure}

Now, unless $x = 0$ or $x = 1$, player $1$ will defect from the truthful strategy and report $\D$ when her true type is $\Up$, because if she does so the allocation $c$ will be implemented by the system operator, which results in her outcome $\III$, hence giving her a value of $34$.
Thus, $x = 0$ or $x = 1$.

If player $2$ has type $\Up$, then the lottery faced by her if she reports $\Up$ is given by
\[
	L_2(\mu_2(\Up, \Up; \cM_0^d, F, \sigma_{-2}^d)) = 
	\{(3x/8, \I); (3(1-x)/8, \II); (1/4, \III); (3(1-x)/8, \IV); (3x/8, \V)\},
\]
where $\sigma_{-2}^d = \sigma_1^d$ is the truthful strategy of player $1$.
If $x = 0$, then the CPT value for player $2$ is given by
\begin{align*}
	&V_2^{\Up}(\{(0, \I); (3/8, \II); (1/4, \III); (3/8, \IV); (0, \V)\}\\
	 &= 
	 -23w(3/8) - 16w(5/8) - 17 \\ 
	&= -33.48.
\end{align*}
If $x = 1$, then the CPT value for player $2$ is given by
\begin{align*}
	&V_2^{\Up}(\{(3/8, \I); (0, \II); (1/4, \III); (0, \IV); (3/8, \V)\} \\
	&= -46w(3/8) - 33w(5/8)\\
	&= -33.48.
\end{align*}
Now, if $x = 0$ or $x = 1$, player $2$ will defect from the truthful strategy and report $\D$ when her true type is $\Up$, because if she does so the allocation $c$ will be implemented by the system operator, which results in her outcome $\III$, hence giving her a value of $-33$.
Thus $x$ cannot be $0$ or $1$, leading to a contradiction.
Thus, $f^*$ cannot be a convex combination of allocation choice functions each of which 
is 
truthfully implementable by a 
non-mediated direct mechanism.
This completes the argument.
\end{example}


\section{Remarks and future directions}
\label{sec: discussion}

Generally in the settings where agents exhibit deviations from expected utility behavior, one would expect that the participating agents do not possess large computational power.
Hence, truthful strategies are especially suitable for such settings in contrast to 
the 
more complicated strategies that are 
permitted by the concept of 
Bayes-Nash equilibrium.
On the other hand, if our participating agents do not possess large computational power, then it is natural to question 
if 
they
have the ability to exhibit strategic behavior, in particular 
the requirement that the strategies form a Bayes-Nash equilibrium (or dominant equilibrium or belief-dominant equilibrium).
However, there can also be agents in the system who do possess large computational power.
Indeed, most of the systems such as online auctions and marketplaces or networked-systems such as transportation networks, Internet routing networks, etc. are comprised of players having varying degrees of computational and strategic abilities.
For example, a firm participating in an online marketplace has the resources to estimate the common prior and other players' strategies through extensive data collection, and thus can develop optimal strategies.
On the other hand, individual agents participating in the same system often lack such resources.
When truthful strategies are in equilibrium, we get the best of both the worlds -- it is easy for the players with limited resources to implement optimal strategies and at the same time there is no incentive for the players with large resources to deviate from these strategies.

Consider the setting when players have independent types, i.e. the common prior $F$ on the type profiles has a product distribution $F = \prod_{i} F_i$.
Let 
$\cM^d = ((\Msg_i)_{i \in [n]}, \mdt, (\Type_i)_{i \in [n]}, h^d)$ 
be a 
direct mediated
mechanism in such a setting.
We note that the lottery induced on the outcome set of player $i$ when she receives a message $\msg_i$, has type $\type_i$, and decides to report $\tilde \type$, is independent of her own type $\type_i$.
This is because her belief $F_{-i}(\cdot|\type_i)$ on the type profiles of her opponents is independent of her type $\type_i$.
With an abuse of notation, let us denote this belief by $F_{-i} \in \Delta(\Type_{-i})$.
Then the lottery induced on the outcome set of player $i$ when she receives a message $\msg_i \in \supp \mdt_i$, and decides to report $\tilde \type$, is given by
\[
	L_i^{\msg_i, \tilde \type_i}(\gamma_i) := \sum_{\msg_{-i}} \mdt_{-i}(\msg_{-i} | \msg_i) \sum_{\type_{-i}} F_{-i}(\type_{-i}) \sum_{\alpha} h^d(\alpha|\msg, \tilde \type, \type_{-i}) \zeta_i(\gamma_i|\alpha), \: \gamma_i \in \Gamma_i.
\]
We will now interpret the message profile as determining the menu of options to be presented to each player.
For example, if the message profile $\msg \in \Msg$ is drawn from the distribution $\mdt$, then player $i$ would be presented with the menu 
comprised of 
lotteries, one for each type $\tilde \type_i \in \Type_i$ of the player.
Let 
\[
	\cal{L}_i(\msg_i)  :=\{L_i^{\msg_i, \tilde \type_i}\}_{\tilde \type_i \in \Type_i},
\]
denote the list of lotteries presented to player $i$ when her message is $\msg_i \in \Msg_i$.
Depending on the player's type, she chooses the lottery that gives her maximum CPT value.
If truthful strategies form an $F$-Bayes-Nash equilibrium, then the lottery $L_i^{\msg_i, \type_i}$ is indeed the best option for a player with type $\type_i$.

In several practical situations, the players are unaware of the type sets of other players $\Type_j, j \neq i$, the allocation set $A$, the allocation-outcome mapping $\zeta$, and the common prior $F$.
It might also be preferable to relieve the players from the burden of knowing the message sets and the mediator distribution $\mdt$.
Note that
the system operator has 
enough knowledge to 
construct the list of lotteries $\cal{L}_i(\msg_i)$ for each player $i$ based on her sampled message $\msg_i$.
Now, using the knowledge of her own type 
$\type_i$,
namely her preferences on the lotteries over her outcome set, player $i$ can select the lottery that is optimal for her from the list $\cal{L}_i(\msg_i)$.
This provides a way to operate the mechanism $\cM^d$ under reasonable assumptions on the players' information.

Further, it is beneficial to limit the complexity of the list $\cal{L}_i(\msg_i)$ presented to the players.
A way to 
do this
would be to limit the size of the list and the complexity of each individual lottery in the list.
The complexity of each individual lottery can be restricted, for example, by limiting the size of the outcome set $\Gamma_i$ and by restricting the probabilities of each outcome to belong to a grid $\{k/K : 0 \leq k \leq K\}$, where $K > 0$ determines the granularity of the grid.
Our framework with separate allocation and outcome sets is helpful in imposing restrictions on the size of the outcome set $\Gamma_i$.
Subsequently, for any 
lottery $L_i \in \Delta(\Gamma_i)$, we can find an approximate lottery $\tilde L_i = \{(p_i(\gamma_i), \gamma_i)\}_{\gamma_i \in \Gamma_i}$ such that $p_i(\gamma_i) \in \{k/K : 0 \leq k \leq K\}$ for all $\gamma_i \in \Gamma_i$.


On the other hand, the size of the list $\cal{L}_i(\msg_i)$ is same as the size of the type set $\Type_i$ in the worst case.
This could make things practically infeasible.
For example, when considering type spaces 
comprised 
of general CPT preferences, it might be 
impossible in practice 
to elicit the probability weighting functions from the agents.
Restricting the type space can lead to inefficient social choice functions.
The mediated mechanism design framework could allow us to limit the size of menu options and at the same time have diversity in the social choice function across different types of the players, facilitated by the messaging stage.
Such multiple communication rounds have been studied under EUT and there is an extensive literature concerning the 
communication 
requirements in mechanism design. (See~\citet{mookherjee2014mechanism} and the references therein. See also 
the literature on 
computational mechanism design~\citep{conitzer2004computational}.)
Given that the non-EUT preferences can reliably be applied only to non-dynamic decision-making, we are especially interested in mechanisms that have a single stage of mediator messages to which the participating agents respond optimally by choosing their best option.
It would be interesting to study the design of mechanisms that optimally elicit CPT preferences under 
communication 
restrictions such as limiting the size of the menu options.
For example, we could consider mechanism designs where the mediated allocation function $h^d$ for a 
direct mediated
mechanism has to satisfy $|\{ \cal{L}_i(\msg_i\} | \leq B$, for all messages $\msg_i$, for some bound $B$.

In this paper, we focused on the mechanism design framework and the revelation principle for agents having CPT preferences. 
It is just the first step towards mechanism design for non-EUT players, with
several interesting directions for future work.

\theendnotes

\appendix
\appendixpage


\section{Proof of the Revelation Principle}
\label{sec: rev_proof}
	We will first consider the revelation principle in the setting of mediated mechanisms.
	This corresponds to statement (i) and a part of statement (ii) of theorem~\ref{thm: rev_prin_1}.
	In this setting we will show that if an allocation choice function $f$
	is implementable in Bayes-Nash equilibrium (resp. belief-dominant equilibrium) 
	by a mediated mechanism
	then it is truthfully implementable in Bayes-Nash equilibrium (resp. belief-dominant equilibrium) by a 
	direct mediated
	mechanism.
	We will then consider the setting of publicly mediated mechanisms and show that if an allocation choice function $f$
	is implementable in dominant equilibrium (resp. belief-dominant equilibrium) 
	by a publicly mediated mechanism 
	then it is truthfully implementable in dominant equilibrium (resp. belief-dominant equilibrium) by a 
direct publicly mediated
	mechanism.
	This will complete the proof of statement (ii) and the remaining part of statement (iii) of theorem~\ref{thm: rev_prin_1}.

	For the first setting, let
	$$\cM = ((\Msg_i)_{i \in [n]}, \mdt, (\Sig_i)_{i \in [n]}, h ),$$ 
	be a mediated mechanism and let
	$\tau$ be a strategy profile that induces $f$ for this mechanism.
	Consider now the 
	direct mediated
	mechanism
$$\cM^d = ((\Msg_i')_{i \in [n]}, \mdt', (\Type_i)_{i \in [n]}, h^d),$$ 
	where the message set is given by 
	\begin{equation}
	\label{eq: def_typeset_rev_proof_mod}
	\Msg_i' := \Msg_i \times \l(\Sig_i \r)^{\Type_i},
	\end{equation}
	with a typical element denoted by 
	\begin{equation}
	\label{eq: msg_prime_def}
		\msg_i' := (\msg_i, (\sig_i^{\type_i'})_{\type_i' \in \Type_i}),
	\end{equation}
	and the mediator distribution $\mdt'$ is given by 
	\begin{equation}
	\label{eq: E_prime_def}
		\mdt'(\msg') := \mdt(\msg) \prod_{i \in [n]} \prod_{\type_i' \in \Type_i} \tau_i\l(\sig_i^{\type_i'} | \msg_i, \type_i'\r) \text{ for all } \msg' \in \Msg'.
	\end{equation}
	The modified mediator messages and the mediator distribution can be interpreted as encapsulating the randomness in the strategies of the players for each of their types into their private messages.

	{
\allowdisplaybreaks
	We now observe that
	\begin{align}
		\label{eq: E'_i_formula}
		\mdt_i'(\msg_i') = \mdt_i(\msg_i) \prod_{\type_i' \in \Type_i} \tau_i\l( \sig_i^{\type_i'} | \msg_i, \type_i' \r),
	\end{align}
	and
		\begin{align*}
		\sum_{\msg_i' \in \Msg_i'} \mdt_i'(\msg_i') = \sum_{\msg' \in \Msg'} \mdt'(\msg') = 1.
	\end{align*}
	Thus, $\mdt' \in \Delta(\Msg')$ is indeed a valid distribution.
	Equation \eqref{eq: E'_i_formula} can be formally proved as follows:
	\begin{align*}
	\mdt_i'(\msg_i') &= \sum_{\msg_{-i}' \in \Msg_{-i}'} \mdt'(\msg_i', \msg_{-i}')\nonumber \\ \nonumber
	&=\sum_{\msg_{-i} \in \Msg_{-i}} \mdt(\msg_i, \msg_{-i})  \sum_{
	\substack{
	(\sig_j^{\type_j'} )_{\type_j' \in \Type_j, j \neq i}\\ \nonumber
	\in \prod_{j \neq i} ( \Sig_j)^{\Type_j}
	}
	}
	\l( \prod_{j \neq i} \prod_{\type_j' \in \Type_j} \tau_j\l(\sig_j^{\type_j'} | \msg_j, \type_j'  \r) \r) 
	\prod_{\type_i' \in \Type_i} \tau_i\l( \sig_i^{\type_i'} | \msg_i, \type_i' \r)\\ \nonumber
	&= \sum_{\msg_{-i} \in \Msg_{-i}} \mdt(\msg_i, \msg_{-i}) \prod_{\type_i' \in \Type_i} \tau_i\l( \sig_i^{\type_i'} | \msg_i, \type_i' \r)
	\l( \prod_{j \neq i} \prod_{\type_j' \in \Type_j} 
	\sum_{\sig_j^{\type_j'} \in \Sig_j}
	\tau_j\l(\sig_j^{\type_j'} | \msg_j, \type_j'  \r) \r)  \\	\nonumber
	&= \sum_{\msg_{-i} \in \Msg_{-i}} \mdt(\msg_i, \msg_{-i}) \prod_{\type_i' \in \Type_i} \tau_i\l( \sig_i^{\type_i'} | \msg_i, \type_i' \r)
	\l( \prod_{j \neq i} \prod_{\type_j' \in \Type_j} 1\r)\nonumber \\ 
	&= \mdt_i(\msg_i) \prod_{\type_i' \in \Type_i} \tau_i\l( \sig_i^{\type_i'} | \msg_i,\type_i' \r).
	\end{align*}

	Let the direct mediated allocation function be given by
	\begin{equation}
	\label{eq: h_prime_def}
		h^d(\msg', \type') := h\l( \msg, \l(\sig_i^{\type_i'}\r)_{i \in [n]}\r) \text{ for all } \msg' \in \Msg', \type' \in \Type.
	\end{equation}
	Note that the construction of the 
	direct mediated
	mechanism is independent of the prior distribution $F$.

	The modified mediator messages and the direct mediated allocation function $h^d$ essentially transfer the randomness in the strategies of the players to the mediator messages, thus allowing each player to simply report her type.
	We observe that the truthful strategies 
	\[
		\tau_i^d(\tilde \type_i|\msg_i', \type_i) = \1\{\tilde \type_i = \type_i\},
	\]
	for all players $i$, implement the allocation choice function $f$ for the direct mediated mechanism $\cM^d$.
	Here is a formal proof.
}

{
\allowdisplaybreaks
Let us compute the distribution on the allocation set induced by the truthful strategy for the direct mediated mechanism.
	For any fixed $\type \in \Type$ and $\alpha \in A$, we have
	\begin{align*}
		&\sum_{\msg' \in \Msg'} \mdt'(\msg') \sum_{\tilde \type \in \Type} \l(\prod_{i \in [n]} \tau_i^d(\tilde \type_i | \msg_i', \type_i) \r) h^d(\alpha | \msg', \tilde \type)\\
		&= \sum_{\msg' \in \Msg'} \mdt'(\msg') \sum_{\tilde \type \in \Type} \l(\prod_{i \in [n]} \1\{\tilde \type_i = \type_i\}\r) h^d(\alpha | \msg', \tilde \type) \\
		& \pushright{\text{ ... because $\tau^d$ is a truthful strategy }}\\
		&= \sum_{\msg' \in \Msg'} \mdt'(\msg')  h^d(\alpha | \msg', \type)\\
		&= \sum_{\msg' \in \Msg'} \mdt(\msg) \l(\prod_{i \in [n]} \prod_{\type_i' \in \Type_i} \tau_i\l(\sig_i^{\type_i'} | \msg_i, \type_i'\r)\r) h^d(\alpha | \msg', \type)\\
		&  \pushright{\text{ ... from \eqref{eq: E_prime_def}}} \\
		&= \sum_{\msg \in \Msg} \mdt(\msg) \sum_{
		\substack{
		(\sig_i^{\type_i'})_{\type_i' \in \Type_i, i \in [n]}\\
		\in \prod_{i \in [n]} \l(\Sig_i\r)^{\Type_i}
		}} \l(\prod_{i \in [n]} \prod_{\type_i' \in \Type_i} \tau_i\l(\sig_i^{\type_i'} | \msg_i, \type_i'\r)\r) h^d(\alpha | \msg', \type)\\
		&\pushright{\text{ ... from \eqref{eq: msg_prime_def} }} \\
		&= \sum_{\msg \in \Msg} \mdt(\msg) \sum_{
		\substack{
		(\sig_i^{\type_i'})_{\type_i' \in \Type_i, i \in [n]}\\
		\in \prod_{i \in [n]} \l(\Sig_i\r)^{\Type_i}
		}}  
		\l(\prod_{i \in [n]} \prod_{\type_i' \in \Type_i} \tau_i\l(\sig_i^{\type_i'} | \msg_i, \type_i'\r)\r) 
		h\l(\alpha \big | \msg, \l(\sig_i^{\type_i}\r)_{i \in [n]}\r)\\
		& \pushright{\text{ ... from \eqref{eq: h_prime_def} }} \\
		&=\sum_{\msg \in \Msg} \mdt(\msg) \sum_{
		\substack{
		(\sig_i^{\type_i})_{i \in [n]}\\
		\in \prod_{i \in [n]} \Sig_i
		}
		} \sum_{
		\substack{
		(\sig_i^{\type_i'})_{\type_i' \neq \type_i, i \in [n]}\\
		\in \prod_{i \in [n]} (\Sig_i)^{\Type_i \back \type_i}
		}
		} 
		\l(\prod_{i \in [n]} \prod_{\type_i' \neq \type_i} \tau_i\l(\sig_i^{\type_i'} | \msg_i, \type_i'\r)\r)\\
		 & \hspace{2in} \times \l( \prod_{i \in [n]} \tau_i\l(\sig_i^{\type_i} | \msg_i, \type_i\r) \r) 
		 h\l(\alpha \big | \msg, \l(\sig_i^{\type_i}\r)_{i \in [n]}\r)\\
		&= \sum_{\msg \in \Msg} \mdt(\msg) \sum_{
		\substack{
		(\sig_i^{\type_i})_{i \in [n]}\\
		\in \prod_{i \in [n]} \Sig_i
		}
		} \l(\prod_{i \in [n]} \tau_i\l(\sig_i^{\type_i} | \msg_i, \type_i\r)\r) 
		h\l(\alpha \big | \msg, \l(\sig_i^{\type_i}\r)_{i \in [n]}\r)\\ 
		& \hspace{2in} \times \sum_{
		\substack{
		(\sig_i^{\type_i'})_{\type_i' \neq \type_i, i \in [n]}\\
		\in \prod_{i \in [n]} (\Sig_i)^{\Type_i \back \type_i}
		}
		} \l(\prod_{i \in [n]} \prod_{\type_i' \neq \type_i} \tau_i\l(\sig_i^{\type_i'} | \msg_i, \type_i'\r)\r)\\
		&=\sum_{\msg \in \Msg} \mdt(\msg) \sum_{
		\substack{
		(\sig_i^{\type_i})_{i \in [n]}\\
		\in \prod_{i \in [n]} \Sig_i
		}
		} \l(\prod_{i \in [n]} \tau_i\l(\sig_i^{\type_i} | \msg_i, \type_i\r)\r) 
		h\l(\alpha \big | \msg, \l(\sig_i^{\type_i}\r)_{i \in [n]}\r) \\
		& \hspace{2in} \times \l(\prod_{i \in [n]} \prod_{\type_i' \neq \type_i} \sum_{\sig_i^{\type_i'} \in \Sig_i} \tau_i\l(\sig_i^{\type_i'} | \msg_i, \type_i'\r)\r)\\
		&= \sum_{\msg \in \Msg} \mdt(\msg) \sum_{
		\substack{
		(\sig_i^{\type_i})_{i \in [n]}\\
		\in \prod_{i \in [n]} \Sig_i
		}
		} \l(\prod_{i \in [n]} \tau_i\l(\sig_i^{\type_i} | \msg_i, \type_i\r)\r) h\l(\alpha \big | \msg, \l(\sig_i^{\type_i}\r)_{i \in [n]}\r) \l(\prod_{i \in [n]} \prod_{\type_i' \neq \type_i} 1\r)\\
		& \pushright{\text{ ... because $\tau_i(\cdot|\msg_i, \type_i') \in \Delta(\Sig_i)$ }}\\
		&= \sum_{\msg \in \Msg} \mdt(\msg) \sum_{\sig \in \Sig} \l(\prod_{i \in [n]} \tau_i\l(\sig_i | \msg_i, \type_i\r)\r) h\l(\alpha \big | \msg, \sig \r)\\
		&= f(\alpha | \type) \mbox{  if $\type \in \supp F$ } \\
		&\pushright{\text{ ... from \eqref{eq: f_induced}.} }
	\end{align*}
	This confirms that the truthful strategy profile implements the social choice function for the direct mediated mechanism $\cM^d$.
}


{
\allowdisplaybreaks
We will now show that if $\tau$ is an $F$-Bayes-Nash equilibrium for $\cM$, then $\tau^d$ is an $F$-Bayes-Nash equilibrium for $\cM^d$.
We will then show that if $\tau$ is a belief-dominant equilibrium for $\cM$, then $\tau^d$ is a belief-dominant equilibrium for $\cM^d$.
To prove these two statements, we first make the following observation concerning the lottery induced over the allocations for player $i$
in the setting of the 
direct mediated
mechanism $\cM^d$, 
when she 
receives the message
$\msg_i':= (\msg_i, (\sig_i^{\type_i'})_{\type_i' \in \Type_i}) \in \supp \mdt_i'$, has 
type 
$\type_i \in \Type_i$, 
has a belief $G_{-i}' \in \Delta(\Type_{-i})$ on the 
opponents' 
type reports (which are the signals of the opponents in this 
direct mediated
mechanism), and decides to report $\tilde \type_i$. 
The lottery induced over the allocations for player $i$ satisfies
\begin{align}
\label{eq: rev_mu_match}
	\mu_i'(\msg_i', \type_i, \tilde \type_i; \cM^d, G_{-i}')
	&:= \sum_{\msg_{-i}' \in \Msg_{-i}'} \mdt_{-i}'(\msg_{-i}' | \msg_i')
	\sum_{\type_{-i} \in \Type_{-i}} G_{-i}'(\type_{-i}) h^d(\msg', \tilde \type_i, \type_{-i})\nn \\
	&= \sum_{\msg_{-i} \in \Msg_{-i}} \mdt_{-i}(\msg_{-i} | \msg_i)
	\sum_{\type_{-i} \in \Type_{-i}} G_{-i}'(\type_{-i}) \nn \\
	&\hspace{1in} \times 
	\sum_{\sig_{-i}}
	\l( \prod_{j \neq i} \tau_j\l(\sig_j | \msg_j, \type_j  \r) \r)
	h\l(\msg, \sig_i^{\tilde \type_i}, \sig_{-i}\r).
\end{align}
We give a formal proof of this in Appendix~\ref{sec: proof_rev_mu_equiv}.
Let us see how this observation helps us prove the two statements above, namely, $\tau^d$ is an equilibrium ($F$-Bayes-Nash or belief-dominant resp.) of $\cM^d$ given that $\tau$ is an equilibrium ($F$-Bayes-Nash or belief-dominant resp.) of $\cM$.

Suppose $F$ is the common prior and $\tau$ is an $F$-Bayes-Nash equilibrium for the mediated mechanism $\cM$.
Let $\msg_i' \in \supp \mdt_i'$ and $\type_i \in \supp F_i$.
From \eqref{eq: E'_i_formula}, we know that $\mdt'_i(\msg_i') > 0$ implies $\mdt_i(\msg_i) > 0$ and $\tau_i(\sig_i^{\type_i'} | \msg_i, \type_i') > 0$, for all $\type_i' \in \Type_i$,
(and in particular, we have $\tau_i(\sig_i^{\type_i} | \msg_i, \type_i) > 0$).
Since $\tau$ is a Bayes-Nash equilibrium for $\cM$, we have
\[
	W_i^{\type_i}\l(\mu_i(\msg_i, \type_i, \sig_i^{\type_i}; \cM, F, \tau_{-i})\r) \geq W_i^{\type_i}\l(\mu_i(\msg_i, \type_i, \tilde \sig_i; \cM, F, \tau_{-i})\r),
\]
for all $\tilde \sig_i \in \Sig_i$.
(Note that $\sig_i^{\type_i} \in \supp \tau_i(\cdot| \msg_i, \type_i)$, $\msg_i \in \supp \mdt_i$, $\type_i \in \supp F_{i}$.)
Taking $G_{-i}' = F_{-i}(\cdot | \type_i)$ in \eqref{eq: rev_mu_match}, we get that
\begin{align}
	\mu_i'(\msg_i', \type_i, \tilde \type_i; \cM^d, F, \tau_{-i}^d) = \mu_i(\msg_i, \type_i, \sig_i^{\tilde \type_i}; \cM, F, \tau_{-i}),
\end{align}
for all $\tilde \type_i \in \Type_i$, and thus,
\begin{align}
\label{eq: IC_ineq_med_BN}
	W_i^{\type_i}\l(\mu_i'(\msg_i', \type_i, \type_i; \cM^d, F, \tau_{-i}^d) \r) &= W_i^{\type_i}\l(\mu_i(\msg_i, \type_i, \sig_i^{\type_i}; \cM, F, \tau_{-i}) \r) \nn \\
	&\geq W_i^{\type_i}\l(\mu_i(\msg_i, \type_i, \sig_i^{\tilde \type_i}; \cM, F, \tau_{-i}) \r)\nn \\
	&= W_i^{\type_i}\l(\mu_i'(\msg_i', \type_i, \tilde \type_i; \cM^d, F, \tau_{-i}^d) \r),
\end{align}
for all $\tilde \type_i \in \Type_i.$
This establishes that the truthful strategy $\tau^d$ is 
an $F$-Bayes-Nash equilibrium 
for $\cM$.

Now suppose $\tau$ is a belief-dominant strategy for $\cM$.
Let $\msg_i' \in \supp \mdt_i'$ and $\type_i \in \Type_i$.
Again, this implies $\mdt_i(\msg_i) > 0$ and $\sig_i^{\type_i} \in \supp \tau_i(\msg_i, \type_i)$.
Corresponding to a belief $G_{-i}' \in \Delta(\Type_{-i})$,
consider the belief $G_{-i} \in \Delta(\Sig_{-i})$ given by
\begin{align}
	G_{-i}(\sig_{-i}) := \sum_{\msg_{-i} \in \Msg_{-i}} \mdt_{-i}(\msg_{-i} | \msg_i)
	\sum_{\type_{-i} \in \Type_{-i}} G_{-i}'(\type_{-i})
	\l( \prod_{j \neq i} \tau_j\l(\sig_j | \msg_j, \type_j  \r) \r)
\end{align}
Then, from \eqref{eq: rev_mu_match}, we have that
\begin{align}
	\mu_i'(\msg_i', \type_i, \tilde \type_i; \cM^d, G_{-i}') = \mu_i(\msg_i, \type_i, \sig_i^{\tilde \type_i}; \cM, G_{-i} )
\end{align}
Noting that 
$\sig_i^{\tilde \type_i}\in \supp \tau_i(\msg_i, \tilde \type_i)$ for all 
$\tilde \type_i \in \Type_i$
 and $\msg_i \in \supp \mdt_i$, and $\tau_i$ being a belief-dominant strategy, we get that
\begin{align}
\label{eq: IC_ineq_med_BD}
	W_i^{\type_i}\l(\mu_i'(\msg_i', \type_i, \type_i; \cM^d, G_{-i}') \r)
	\geq W_i^{\type_i}\l(\mu_i'(\msg_i', \type_i, \tilde \type_i; \cM^d, G_{-i}') \r)
\end{align}
for all $\tilde \type_i \in \Type_i$.
Thus, the truthful strategy $\tau^d$ is a belief-dominant strategy for $\cM^d$.


}

This completes the proof of statement (i) in theorem~\ref{thm: rev_prin_1} and part of statement (iii) corresponding to mediated mechanisms.
We now consider the setting of publicly mediated mechanisms and establish the rest of the theorem.

Let
\[
	\cM_* = (\Msg_*, \mdt_*, (\Sig_i)_{i \in [n]}, h_*)
\]
be a publicly mediated mechanism and 
for each player $i$ let
$\tau_i: \Msg_* \times \Type_i \to \Delta(\Sig_i)$ be her strategy such that the strategy profile $\tau$ induces 
the allocation choice function
$f$ for this mechanism.
We now consider the 
direct publicly mediated
mechanism
\[
	\cM_*^d := (\Msg_*', \mdt_*', (\Type_i)_{i \in [n]}, h_*^d),
\]
where the message set is
given by
	$$\Msg_*' := \Msg_* \times \prod_{i = 1}^n \l(\Sig_i \r)^{\Type_i},$$
	with a typical element denoted by 
	\begin{equation}
	\label{eq: msg_prime_def*}
		\msg_*' := (\msg_*, (\sig_i^{\type_i'})_{\type_i' \in \Type_i, i \in [n]}),
	\end{equation}
	and the mediator distribution $\mdt_*'$ is
	given by
	\begin{equation}
	\label{eq: E_prime_def*}
		\mdt_*'(\msg_*') := \mdt_*(\msg_*) \prod_{i \in [n]} \prod_{\type_i' \in \Type_i} \tau_i\l(\sig_i^{\type_i'} | \msg_i, \type_i'\r) \text{ for all } \msg' \in \Msg'.
	\end{equation}
	Similar to the previous setting, here the modified mediator messages and the mediator distribution can be interpreted as encapsulating the randomness in the strategies of the players for each of their types into the public messages.
	We can similarly verify that $\mdt_*'$ is indeed a probability distribution on $\Msg_*'$.
	The direct mediated allocation function $h_*^d$ in the direct publicly mediated mechanism
	$\cM^d$ is
	given by
	\begin{equation}
	\label{eq: h_prime_def*}
		h_*^d(\msg_*', \type') := h_*\l( \msg_*, \l(\sig_i^{\type_i'}\r)_{i \in [n]}\r) \text{ for all } 
		\msg_*' \in \Msg_*',
		\type' \in \Type.
	\end{equation}
	We can similarly verify that the truthful strategies 
	\[
		\tau^d(\msg_*', \type_i) = \type_i 
	\]
	implement the allocation choice function $f$ for $\cM_*^d$.

	
	Fix $\msg_*' \in \supp \mdt_*'$.
	Note that
	\begin{align}
	\label{eq: dominant_allocation_direct_equiv}
		h_*^d(\msg_*', \tilde \type_i, \type_{-i}) = h_*(\msg_*, \sig_i^{\tilde \type_i}, (\sig_j^{\type_j})_{j \neq i}),
	\end{align}
	for all $\tilde \type_i \in \Type_i$.
	From \eqref{eq: E_prime_def*}, we have $\msg_* \in \supp \mdt_*$ and $\sig_i^{\type_i} \in \supp \tau_i(\msg_*, \type_i)$
	for all $\type_i \in \Type_i$.

	Now suppose $\tau$ is a dominant equilibrium for $\cM_*$.
	The lottery induced over the allocations for player $i$ when she receives a publicly mediated message $\msg_*'$, has type $\type_i$, believes that the 
	opponents
	are reporting $\type_{-i}$, and decides to report $\tilde \type_i$ is given by
	\begin{align}
	\label{eq: mu_prime_pub_med_mech_prime}
		\mu_i'(\msg_*', \type_i, \tilde \type_i; \cM_*^d, \type_{-i} ) = h_*^d(\msg_*', \tilde \type_i, \type_{-i}).
	\end{align}
	We get this from \eqref{eq: induce_lot_med_dom_eq} by considering the special case of publicly mediated mechanisms.
	From \eqref{eq: dominant_allocation_direct_equiv}, we get that this is equal to the lottery induced over the allocations for player $i$ when she receives a publicly mediated message $\msg_*$, has type $\type_i$, believes that the 
	opponents
	are reporting $\sig_j^{\type_j}, j \neq i$, and decides to report $\sig_i^{\tilde \type}$, namely,
	\begin{align*}
		\mu_i(\msg_*, \type_i, \sig_i^{\tilde \type}; \cM_*, (\sig_j^{\type_j})_{j \neq i} ) = h_*(\msg_*, \sig_i^{\tilde \type}, (\sig_j^{\type_j})_{j \neq i}).
	\end{align*}
	Since $\tau_i$ is a dominant strategy, $\msg_* \in \supp \mdt_*$, and $\sig_i^{\type_i} \in \supp \tau_i(\msg_*, \type_i)$, we have
	\begin{align*}
		W_i^{\type_i}(\mu_i(\msg_*, \type_i, \sig_i^{\type_i}; \cM_*, (\sig_j^{\type_j})_{j \neq i} )) \geq W_i^{\type_i}(\mu_i(\msg_*, \type_i,  \tilde \sig_i; \cM_*, (\sig_j^{\type_j})_{j \neq i} )),
	\end{align*}
	for all $\tilde \sig_i \in \Sig_i$.
	Hence, we have
	\begin{align*}
		W_i^{\type_i}(\mu_i'(\msg_*', \type_i, \type_i; \cM_*^d, \type_{-i} )) \geq W_i^{\type_i}(\mu_i'(\msg_*', \type_i, \tilde \type_i; \cM_*^d, \type_{-i} )),
	\end{align*}	
	for all $\tilde \type_i \in \Type_i$.
	Thus, $\tau^d$ is a dominant equilibrium of $\cM_*^d$.

	Now suppose that $\tau$ is a belief-dominant equilibrium for $\cM_*$.
	Consider the fixed message
	$$\msg_*' = (\msg_*, (\sig_i^{\type_i'})_{\type_i' \in \Type_i, i \in [n]}) \in \supp \mdt_*',$$ 
	as before.
	Corresponding to a belief $G_{-i}' \in \Delta(\Type_{-i})$, consider $G_{-i, *} \in \Delta(\Sig_{-i})$ 
	given by
	\begin{align}
	\label{eq: def_belief_dom_belief_equiv}
		G_{-i, *}(\tilde \sig_{-i}) := \sum_{
		\substack{
			\type_{-i} \in \Type_{-i}\\
			\text{s.t.} \sig_j^{\type_j} = \tilde \sig_{-i}, \forall j \neq i
		}
		} G_{-i}'(\type_{-i}),
	\end{align}
	for all $\tilde \sig_{-i} \in \Sig_{-i}$, where $\sig_j^{\type_j}$ are the signals corresponding to the types as defined by the message $\msg_*'$.

{
\allowdisplaybreaks	
    As observed in equation \eqref{eq: mu_prime_pub_med_mech_prime}, the lottery induced over the allocations for player $i$, when she receives message $\msg_*'$, has type $\type_i$, believes that the opponents' are reporting $\type_{-i}$, and decides to report $\tilde \type_i$ is given by $h^d_*(\msg_*', \tilde \type_i, \type_{-i})$.
    Now suppose that she has belief 
    $G_{-i}'$ on her opponents' type report instead.
    Then, the induced lottery over the allocations for player $i$ is given by
	\begin{align*}
		\mu_{i}'(\msg_*', \type_i, \tilde \type_i; \cM_*^d, G_{-i}') &= \sum_{\type_{-i} \in \Type_{-i}} G_{-i}'(\type_{-i}) h_*^d(\msg_*', \tilde \type_i, \type_{-i})\\
		&= \sum_{\type_{-i} \in \Type_{-i}} G_{-i}'(\type_{-i}) h_*(\msg_*, \sig_i^{\tilde \type_i}, (\sig_j^{\type_j})_{j \neq i})\\
		& \pushright{\text{ ... from \eqref{eq: dominant_allocation_direct_equiv} }}\\
		&= \sum_{\tilde \sig_{-i} \in \Sig_{-i}} h_*(\msg_*, \sig_i^{\tilde \type_i}, \tilde \sig_{-i}) \sum_{
		\substack{
			\type_{-i} \in \Type_{-i}\\
			\text{s.t.} \sig_j^{\type_j} = \tilde \sig_{-i}, \forall j \neq i
		}
		} G_{-i}'(\type_{-i})\\
		\red
		&=  \sum_{\tilde \sig_{-i} \in \Sig_{-i}} h_*(\msg_*, \sig_i^{\tilde \type_i}, \tilde \sig_{-i}) G_{-i, *}(\tilde \sig_{-i})\\
		\black
		& \pushright{\text{ ... from \eqref{eq: def_belief_dom_belief_equiv}  }}\\
		&= \mu_i(\msg_*, \type_i, \sig_i^{\tilde \type_i}; \cM_*, G_{-i, *} ).
	\end{align*}

Since $\tau$ is a belief-dominant equilibrium, $\msg_* \in \supp \mdt_*$, and $\sig_i^{\type_i} \in \supp \tau_i(\msg_*, \type_i)$
for all $\type_i \in \Type_i$,
we have 
\begin{align*}
		W_i^{\type_i}(\mu_i(\msg_*, \type_i, \sig_i^{\type_i}; \cM_*, G_{-i, *}) \geq W_i^{\type_i}(\mu_i(\msg_*, \type_i, \tilde \sig_i; \cM_*, G_{-i, *}),
	\end{align*}
	for all $\tilde \sig_i \in \Sig_i$.
	Hence, we have
	\begin{align*}
		W_i^{\type_i}(\mu_i'(\msg_*', \type_i, \type_i; \cM_*^d, G_{-i}' )) \geq W_i^{\type_i}(\mu_i'(\msg_*', \type_i, \tilde \type_i; \cM_*^d, G_{-i}' )),
	\end{align*}	
	for all $\tilde \type_i \in \Type_i$.
	Thus, $\tau^d$ is a belief-dominant equilibrium of $\cM_*^d$.
}


This completes the proof of the theorem.


\section{Proof of \eqref{eq: rev_mu_match}}
\label{sec: proof_rev_mu_equiv}

{
\allowdisplaybreaks
Let us recall the first setting considered in Appendix~\ref{sec: rev_proof}. 
We have
a mediated mechanism
	$$\cM = ((\Msg_i)_{i \in [n]}, \mdt, (\Sig_i)_{i \in [n]}, h ),$$ 
and a corresponding strategy profile $\tau$.
We had constructed a 
direct mediated
mechanism
$$\cM^d = ((\Msg_i')_{i \in [n]}, \mdt', (\Type_i)_{i \in [n]}, h^d),$$ 
given by \eqref{eq: def_typeset_rev_proof_mod}, \eqref{eq: msg_prime_def}, \eqref{eq: E_prime_def}, and \eqref{eq: h_prime_def}. 
We are interested in the situation when player $i$ receives message
$\msg_i':= (\msg_i, (\sig_i^{\type_i'})_{\type_i' \in \Type_i}) \in \supp \mdt_i'$, has 
type $\type_i \in \Type_i$, 
and belief $G_{-i}' \in \Delta(\Type_{-i})$ on the 
opponents' 
type reports, and decides to report $\tilde \type_i$.
Since $\mdt'(\msg_i') > 0$ by assumption, we have
\begin{align*}
	&\sum_{\msg_{-i}' \in \Msg_{-i}'} \mdt_{-i}'(\msg_{-i}' | \msg_i')
	\sum_{\type_{-i} \in \Type_{-i}} G_{-i}'(\type_{-i}) h^d(\msg', \tilde \type_i, \type_{-i})\\
	&= \frac{\sum_{\msg_{-i}' \in \Msg_{-i}'} \mdt'(\msg_i', \msg_{-i}')
	\sum_{\type_{-i} \in \Type_{-i}} G_{-i}'(\type_{-i}) h^d(\msg', \tilde \type_i, \type_{-i})}
	{\sum_{\msg_{-i}' \in \Msg_{-i}'} \mdt'(\msg_i', \msg_{-i}')}.
\end{align*}
Let the denominator be denoted by
$$C_1 := \sum_{\msg_{-i}' \in \Msg_{-i}'} \mdt'(\msg_i', \msg_{-i}') = \mdt_i'(\msg_i').$$ 
We now focus on the numerator, to get
\begin{align*}
	&\sum_{\msg_{-i}' \in \Msg_{-i}'} \mdt'(\msg_i', \msg_{-i}')
	\sum_{\type_{-i} \in \Type_{-i}} G_{-i}'(\type_{-i}) h^d(\msg', \tilde \type_i, \type_{-i})\\
	=& \sum_{\msg_{-i} \in \Msg_{-i}} \mdt(\msg_i, \msg_{-i})
	 \sum_{
	\substack{
	(\sig_j^{\type_j'} )_{\type_j' \in \Type_j, j \neq i}\\
	\in \prod_{j \neq i} ( \Sig_j)^{\Type_j}
	}
	}
	\l( \prod_{j \neq i} \prod_{\type_j' \in \Type_j} \tau_j\l(\sig_j^{\type_j'} | \msg_j, \type_j'  \r) \r) \prod_{\type_i' \in \Type_i} \tau_i\l( \sig_i^{\type_i'} | \msg_i,  \type_i' \r) \\
	& \hspace{2in} \times \sum_{\type_{-i} \in \Type_{-i}} G_{-i}'(\type_{-i}) h^d(\msg', \tilde \type_i, \type_{-i})\\
	&\pushright{\text{ ... from \eqref{eq: E_prime_def}}}\\
	=& \sum_{\msg_{-i} \in \Msg_{-i}} \mdt(\msg_i, \msg_{-i})
	 \sum_{
	\substack{
	(\sig_j^{\type_j'} )_{\type_j' \in \Type_j, j \neq i}\\
	\in \prod_{j \neq i} ( \Sig_j)^{\Type_j}
	}
	}
	\l( \prod_{j \neq i} \prod_{\type_j' \in \Type_j} \tau_j\l(\sig_j^{\type_j'} | \msg_j, \type_j'  \r) \r) 
	\l(\prod_{\type_i' \in \Type_i} \tau_i\l( \sig_i^{ \type_i'} | \msg_i, \type_i' \r)\r) \\
	& \hspace{2in} \times \sum_{\type_{-i} \in \Type_{-i}} G_{-i}'(\type_{-i}) h\l(\msg_i, \msg_{-i}, \sig_i^{\tilde \type_i}, \l(\sig_j^{\type_j}\r)_{j \neq i}\r)\\
	&\pushright{\text{ ... from \eqref{eq: h_prime_def}}}\\
	=& \l(\prod_{\type_i' \in \Type_i} \tau_i\l( \sig_i^{\type_i'} | \msg_i, \type_i' \r) \r)
	\sum_{\msg_{-i} \in \Msg_{-i}} \mdt(\msg_i, \msg_{-i})
	 \sum_{
	\substack{
	(\sig_j^{\type_j'} )_{\type_j' \in \Type_j, j \neq i}\\
	\in \prod_{j \neq i} ( \Sig_j)^{\Type_j}
	}
	}
	\l( \prod_{j \neq i} \prod_{\type_j' \in \Type_j} \tau_j\l(\sig_j^{\type_j'} | \msg_j, \type_j'  \r) \r)  \\
	& \hspace{2in} \times \sum_{\type_{-i} \in \Type_{-i}} G_{-i}'(\type_{-i}) h\l(\msg_i, \msg_{-i}, \sig_i^{\tilde \type_i}, \l(\sig_j^{\type_j}\r)_{j \neq i}\r)
\end{align*}
Let
\[
	C_2 := \prod_{\type_i' \in \Type_i} \tau_i\l( \sig_i^{\type_i'} | \msg_i, \type_i' \r).
\]
We have,
\begin{align*}
	& \sum_{\msg_{-i} \in \Msg_{-i}} \mdt(\msg_i, \msg_{-i})
	 \sum_{
	\substack{
	(\sig_j^{\type_j'} )_{\type_j' \in \Type_j, j \neq i}\\
	\quad \in \prod_{j \neq i} ( \Sig_j)^{\Type_j}
	}
	}
	\l( \prod_{j \neq i} \prod_{\type_j' \in \Type_j} \tau_j\l(\sig_j^{\type_j'} | \msg_j,  \type_j'  \r) \r)  \\
	& \hspace{2in} \times \sum_{\type_{-i} \in \Type_{-i}} G_{-i}'(\type_{-i}) h\l(\msg_i, \msg_{-i}, \sig_i^{\tilde \type_i}, \l(\sig_j^{\type_j}\r)_{j \neq i}\r)\\
	&= \sum_{\msg_{-i} \in \Msg_{-i}} \mdt(\msg_i, \msg_{-i})
	\sum_{\type_{-i} \in \Type_{-i}} G_{-i}'(\type_{-i}) 
	 \sum_{
	\substack{
	(\sig_j^{\type_j'} )_{\type_j' \in \Type_j, j \neq i}\\
	\quad \in \prod_{j \neq i} ( \Sig_j)^{\Type_j}
	}
	}
	\l( \prod_{j \neq i} \prod_{\type_j' \in \Type_j} \tau_j\l(\sig_j^{\type_j'} | \msg_j, \type_j'  \r) \r)  \\
	& \hspace{2in} \times 
	h\l(\msg_i, \msg_{-i}, \sig_i^{\tilde \type_i}, \l(\sig_j^{\type_j}\r)_{j \neq i}\r)\\
	&= \sum_{\msg_{-i} \in \Msg_{-i}} \mdt(\msg_i, \msg_{-i})
	\sum_{\type_{-i} \in \Type_{-i}} G_{-i}'(\type_{-i}) \\
	&\hspace{1in} \times
	\sum_{
	\substack{
	(\sig_j^{\type_j})_{j \neq i} \\
	\quad \in (\Sig_j)_{j \neq i}
	}
	}
	 \sum_{
	\substack{
	(\sig_j^{\type_j'} )_{\type_j' \in \Type_j \back  \type_j, j \neq i}\\
	\quad \in \prod_{j \neq i} ( \Sig_j)^{\Type_j \back \type_j}
	}
	}
	\l( \prod_{j \neq i} \prod_{\type_j' \in \Type_j \back \type_j} \tau_j\l(\sig_j^{\type_j'} | \msg_j, \type_j'  \r) \r)  \\
	& \hspace{2in} \times 
	\l( \prod_{j \neq i} \tau_j\l(\sig_j^{\type_j} | \msg_j, \type_j  \r) \r)
	h\l(\msg_i, \msg_{-i}, \sig_i^{\tilde \type_i}, \l(\sig_j^{\type_j}\r)_{j \neq i}\r)\\
	&= \sum_{\msg_{-i} \in \Msg_{-i}} \mdt(\msg_i, \msg_{-i})
	\sum_{\type_{-i} \in \Type_{-i}} G_{-i}'(\type_{-i}) \\
	&\hspace{1in} \times 
	\sum_{
	\substack{
	(\sig_j^{\type_j})_{j \neq i} \\
	\quad \in (\Sig_j)_{j \neq i}
	}
	}
	\l( \prod_{j \neq i} \tau_j\l(\sig_j^{\type_j} | \msg_j, \type_j  \r) \r)
	h\l(\msg_i, \msg_{-i}, \sig_i^{\tilde \type_i}, \l(\sig_j^{\type_j}\r)_{j \neq i}\r) \\
	& \hspace{2in} \times 
	 \sum_{
	\substack{
	(\sig_j^{\type_j'} )_{\type_j' \in \Type_j \back \type_j, j \neq i}\\
	\quad \in \prod_{j \neq i} ( \Sig_j)^{\Type_j \back \type_j}
	}
	}
	\l( \prod_{j \neq i} \prod_{\type_j' \in \Type_j \back \type_j} \tau_j\l(\sig_j^{\type_j'} | \msg_j, \type_j'  \r) \r)  \\
	&= \sum_{\msg_{-i} \in \Msg_{-i}} \mdt(\msg_i, \msg_{-i})
	\sum_{\type_{-i} \in \Type_{-i}} G_{-i}'(\type_{-i}) \\
	&\hspace{1in} \times 
	\sum_{
	\substack{
	(\sig_j^{\type_j})_{j \neq i} \\
	\quad \in (\Sig_j)_{j \neq i}
	}
	}
	\l( \prod_{j \neq i} \tau_j\l(\sig_j^{\type_j} | \msg_j, \type_j  \r) \r)
	h\l(\msg_i, \msg_{-i}, \sig_i^{\tilde \type_i}, \l(\sig_j^{\type_j}\r)_{j \neq i}\r) \\
	& \hspace{2in} \times 
	\l( \prod_{j \neq i} \prod_{\type_j' \in \Type_j \back \type_j} 
	\sum_{\sig_j^{\type_j'} \in \Sig_j}
	\tau_j\l(\sig_j^{\type_j'} | \msg_j, \type_j'  \r) \r)  \\
	&= \sum_{\msg_{-i} \in \Msg_{-i}} \mdt(\msg_i, \msg_{-i})
	\sum_{\type_{-i} \in \Type_{-i}} G_{-i}'(\type_{-i}) \\
	&\hspace{1in} \times 
	\sum_{
	\substack{
	(\sig_j^{\type_j})_{j \neq i} \\
	\quad \in (\Sig_j)_{j \neq i}
	}
	}
	\l( \prod_{j \neq i} \tau_j\l(\sig_j^{\type_j} | \msg_j, \type_j  \r) \r)
	h\l(\msg_i, \msg_{-i}, \sig_i^{\tilde \type_i}, \l(\sig_j^{\type_j}\r)_{j \neq i}\r) \\
	& \hspace{2in} \times 
	\l( \prod_{j \neq i} \prod_{\tilde \type_j' \in \Type_j \back \type_j} 1\r)  \\
	&= \sum_{\msg_{-i} \in \Msg_{-i}} \mdt(\msg_i, \msg_{-i})
	\sum_{\type_{-i} \in \Type_{-i}} G_{-i}'(\type_{-i}) \\
	&\hspace{1in} \times 
	\sum_{
	\substack{
	(\sig_j^{\type_j})_{j \neq i} \\
	\quad \in (\Sig_j)_{j \neq i}
	}
	}
	\l( \prod_{j \neq i} \tau_j\l(\sig_j^{\type_j} | \msg_j, \type_j  \r) \r)
	h\l(\msg_i, \msg_{-i}, \sig_i^{\tilde \type_i}, \l(\sig_j^{\type_j}\r)_{j \neq i}\r) \\
	&= \sum_{\msg_{-i} \in \Msg_{-i}} \mdt(\msg_i, \msg_{-i})
	\sum_{\type_{-i} \in \Type_{-i}} G_{-i}'(\type_{-i}) \\
	&\hspace{1in} \times 
	\sum_{\sig_{-i}}
	\l( \prod_{j \neq i} \tau_j\l(\sig_j | \msg_j, \type_j  \r) \r)
	h\l(\msg_i, \msg_{-i}, \sig_i^{\tilde \type_i}, \sig_{-i}\r) \\
	&= \mdt_i(\msg_i) \sum_{\msg_{-i} \in \Msg_{-i}} \mdt_{-i}(\msg_{-i} | \msg_i)
	\sum_{\type_{-i} \in \Type_{-i}} G_{-i}'(\type_{-i}) \\
	&\hspace{1in} \times 
	\sum_{\sig_{-i}}
	\l( \prod_{j \neq i} \tau_j\l(\sig_j | \msg_j, \type_j  \r) \r)
	h\l(\msg_i, \msg_{-i}, \sig_i^{\tilde \type_i}, \sig_{-i}\r).
\end{align*}
We recall that ${\mdt_i(\msg_i) C_2}/{C_1} = 1$ from \eqref{eq: E'_i_formula}, and hence, we get \eqref{eq: rev_mu_match}.
}

\section{Outcome sets can be identified with the allocation set under EUT}
\label{sec: new_appendix}

{
\allowdisplaybreaks

Consider a setting in which all the players have EUT preferences for all their types.
For this restricted setting, 
we will now 
construct
an environment 
\[
	\cal{E}' := \l([n], (\Type_i')_{i \in [n]}, A, (\Gamma_i')_{i \in [n]}, \zeta'\r),
\]
that we call the reduced environment
corresponding to the environment 
(as defined in \eqref{eq: def_environment})
\[
	\cal{E} := \l([n], (\Type_i)_{i \in [n]}, A, (\Gamma_i)_{i \in [n]}, \zeta\r).
\]

From~\eqref{eq: W_EUT_linear}, we observe that,
since we are dealing with EUT preferences,
the utility function $W_i^{\type_i}$ is completely determined by the values $u_i^{\type_i}(\alpha), \forall \alpha \in A$.
Suppose the mechanism designer models the outcome 
set of each player $i$
by $\Gamma_i' = A$
instead of the true outcome set $\Gamma_i$,
with the trivial allocation-outcome mapping $\zeta_i'$ instead of 
the original allocation-outcome mapping $\zeta_i$. 
Let $\zeta'$ denote 
the product
of the trivial allocation-outcome mappings $\zeta_i', i \in [n]$.
Corresponding to a type $\type_i \in \Type_i$ for player $i$, 
the mechanism designer models her type by $\type_i'$, 
which is 
characterized by the utility function $u_i^{\type_i} : \Gamma_i' \to \bbR$ 
as defined in \eqref{eq: def_u_of_alpha}.
Since the players are assumed to have EUT preferences, the probability weighting functions 
under each type $\type_i'$
are modeled to be $w_i^\pm(p) = p, \forall p \in [0,1]$.
Let $\Type_i'$ denote the set comprised of all types $\type_i'$ corresponding to the types $\type_i \in \Type_i$.
Let $\cal{T}_i : \Type_i \to \Type_i'$ denote the function for this correspondence.
Suppose the mechanism designer treats the environment as if given by
\[
	\cal{E}' := \l([n], (\Type_i')_{i \in [n]}, A, (\Gamma_i')_{i \in [n]}, \zeta'\r).
\]

Let $\Type' := \prod_{i} \Type_i'$.
Let $\cal{T} : \Type \to \Type'$ denote the product transformation defined by the functions $\cal{T}_i, i \in [n]$.
Notice that the function $\cal{T}_i$ is a bijection 
since, as pointed out earlier, 
even if $u_i^{\type_i} = u_i^{\tilde \type_i}$ for some $\type_i \neq \tilde \type_i$, we will treat $\cal{T}_i(\type_i)$ and $\cal{T}_i(\tilde \type_i)$ as different elements of $\Type_i'$.
%
For any prior $F \in \Delta(\Type)$, let $F' \in \Delta(\Type')$ be the corresponding prior induced by the bijection $\cal{T}$.

Note that, for any player $i$, having any type $\type_i$, and any lottery $\mu \in \Delta(A)$, we have
\begin{equation}
\label{eq: obs_W_EUT_same}
	W_i^{\type_i}(\mu) = W_i^{\cal{T}_i(\type_i) }(\mu).
\end{equation}
(Here, $W_i^{\cal{T}_i(\type_i)}$ should be interpreted as the utility function for player $i$ with type $\type_i' = \cal{T}_i(\type_i)$ corresponding to the reduced environment $\cal{E}'$.)
Let $f' : \Type' \to A$ be an allocation choice function 
that is implementable in $F'$-Bayes-Nash equilibrium $\sigma' := (\sigma_i')_{i \in [n]}$ (where $\sigma_i' : \Type_i' \to \Delta(\Sig_i')$) for
the mechanism
\[
	\cM_0 = ((\Sig_i)_{i \in [n]}, h_0).
\]

Now suppose the system operator uses the same mechanism $\cM_0$ in environment $\cal{E}$.
Consider the allocation choice function $f: \Type \to \Delta(A)$ given by
\[
	f(\type) = f'(\cal{T}(\type)).
\]
For each player $i$, consider the strategy $\sigma_i: \Type_i \to \Delta(\Sig_i)$ given by
\[
	\sigma_i(\type_i) = \sigma_i'(\cal{T}_i(\type_i)).
\]
Similar to \eqref{eq: belief_bayes_induced}, for any $\type_i' \in \supp F_i'$ and signal $\sig_i$, let 
\begin{equation}
\label{eq: belief_bayes_induced_EUTreduced}
	\mu_i'(\type_i', \sig_i; \cM_0, F', \sigma'_{-i})  
	:= \sum_{\type_{-i} \in \Type_{-i}} F'_{-i}(\type_{-i}|\type_i) \sum_{\sig_{-i} \in \Sig_{-i}} \prod_{j \neq i} \sigma_j(\sig_j | \type_j') h_0(\sig),
\end{equation}
be the belief of player $i$ on the allocation set corresponding to the reduced environment $\cal{E}'$.
Note that 
\[
	\mu_i(\type_i, \sig_i; \cM_0, F, \sigma_{-i}) = \mu_i'(\cal{T}_i(\type_i), \sig_i; \cM_0, F', \sigma'_{-i}).
\]
From observation~\eqref{eq: obs_W_EUT_same} and the definition of $F$-Bayes-Nash equilibrium in \eqref{eq: Bayes_Nash_W_def} and \eqref{eq: def_BNE}, we get that the allocation choice function $f$ is implementable in $F$-Bayes-Nash equilibrium by the mechanism $\cM_0$ with the equilibrium strategy $\sigma$.

On the other hand, suppose we have an allocation choice function $f: \Type \to \Delta(A)$.
Consider the corresponding allocation choice function $f':\Type' \to \Delta(A)$ given by
\[
	f'(\type') = f(\cal{T}^{-1}(\type)).
\]
We now observe that if $f$ is implementable in $F$-Bayes-Nash equilibrium by a mechanism $\cM_0$ and an $F$-Bayes-Nash equilibrium $\sigma$, then so is $f'$ by the same mechanism $\cM_0$ and the $F'$-Bayes-Nash equilibrium $\sigma'$ comprised of
\[
	\sigma_i'(\type_i') = \sigma_i(\cal{T}^{-1}_i(\type_i')),
\]
for all $i \in [n], \type_i' \in \Type_i'$.

We can similarly show that if $f'$ is implementable in dominant (resp. belief-dominant) equilibrium by a mechanism $\cM_0$ with the equilibrium strategy 
profile
$\sigma'$ for the reduced environment $\cal{E}'$, then so is $f$ by the same mechanism $\cM_0$ with the corresponding equilibrium strategy 
profile
$\sigma$ for the environment $\cal{E}$, and vice versa.
%


Hence, under EUT, from the mechanism designer's point of view, it is enough to model the types of player $i$
by setting the outcome set $\Gamma_i' = A$, assuming the trivial allocation-outcome mapping $\zeta_i'$, and the types $\type_i' \in \Type_i'$. 
\black
}

\bibliographystyle{abbrvnat} 
\bibliography{Bib_Database}

\begin{thebibliography}{32}
\providecommand{\natexlab}[1]{#1}
\providecommand{\url}[1]{\texttt{#1}}
\expandafter\ifx\csname urlstyle\endcsname\relax
  \providecommand{\doi}[1]{doi: #1}\else
  \providecommand{\doi}{doi: \begingroup \urlstyle{rm}\Url}\fi

\bibitem[Allais(1953)]{allais1953comportement}
M.~Allais.
\newblock Le comportement de l'homme rationnel devant le risque: critique des
  postulats et axiomes de l'{\'e}cole am{\'e}ricaine.
\newblock \emph{Econometrica: Journal of the Econometric Society}, pages
  503--546, 1953.

\bibitem[Aumann(1974)]{aumann1974subjectivity}
R.~J. Aumann.
\newblock Subjectivity and correlation in randomized strategies.
\newblock \emph{Journal of mathematical Economics}, 1\penalty0 (1):\penalty0
  67--96, 1974.

\bibitem[Aumann(1987)]{aumann1987correlated}
R.~J. Aumann.
\newblock Correlated equilibrium as an expression of {B}ayesian rationality.
\newblock \emph{Econometrica: Journal of the Econometric Society}, pages 1--18,
  1987.

\bibitem[Bergemann and Morris(2005)]{bergemann2005robust}
D.~Bergemann and S.~Morris.
\newblock Robust mechanism design.
\newblock \emph{Econometrica}, pages 1771--1813, 2005.

\bibitem[B{\"o}rgers and Krahmer(2015)]{borgers2015introduction}
T.~B{\"o}rgers and D.~Krahmer.
\newblock \emph{An introduction to the theory of mechanism design}.
\newblock Oxford University Press, USA, 2015.

\bibitem[B{\"o}rgers and Oh(2011)]{borgers2011common}
T.~B{\"o}rgers and T.~Oh.
\newblock Common prior type spaces in which payoff types and belief types are
  independent.
\newblock \emph{polar}, 2011\penalty0 (2011b), 2011.

\bibitem[Chateauneuf and Wakker(1999)]{chateauneuf1999axiomatization}
A.~Chateauneuf and P.~Wakker.
\newblock An axiomatization of cumulative prospect theory for decision under
  risk.
\newblock \emph{Journal of Risk and Uncertainty}, 18\penalty0 (2):\penalty0
  137--145, 1999.

\bibitem[Chew(1985)]{chew1985implicit}
S.-h. Chew.
\newblock \emph{Implicit Weighted and Semi-Weighted Utility Theories,
  M-Estimators, and Non-Demand Revelation of Second-Price Auctions for an
  Uncertain Auctioned Object}.
\newblock Johns Hopkins Univ., Department of Political Economy, 1985.

\bibitem[Chew(1989)]{chew1989axiomatic}
S.~H. Chew.
\newblock Axiomatic utility theories with the betweenness property.
\newblock \emph{Annals of operations Research}, 19\penalty0 (1):\penalty0
  273--298, 1989.

\bibitem[Conitzer and Sandholm(2004)]{conitzer2004computational}
V.~Conitzer and T.~Sandholm.
\newblock Computational criticisms of the revelation principle.
\newblock In \emph{Proceedings of the 5th ACM conference on Electronic
  commerce}, pages 262--263, 2004.

\bibitem[Dekel(1986)]{dekel1986axiomatic}
E.~Dekel.
\newblock An axiomatic characterization of preferences under uncertainty:
  Weakening the independence axiom.
\newblock \emph{Journal of Economic theory}, 40\penalty0 (2):\penalty0
  304--318, 1986.

\bibitem[Fishburn and Kochenberger(1979)]{fishburn1979two}
P.~C. Fishburn and G.~A. Kochenberger.
\newblock Two-piece von {N}eumann-{M}orgenstern utility functions.
\newblock \emph{Decision Sciences}, 10\penalty0 (4):\penalty0 503--518, 1979.

\bibitem[Harsanyi(1967)]{harsanyi1967games}
J.~C. Harsanyi.
\newblock Games with incomplete information played by ``{B}ayesian'' players,
  {I}--{III} part i. the basic model.
\newblock \emph{Management science}, 14\penalty0 (3):\penalty0 159--182, 1967.

\bibitem[Kahneman and Tversky(2013)]{kahneman2013prospect}
D.~Kahneman and A.~Tversky.
\newblock Prospect theory: An analysis of decision under risk.
\newblock In \emph{Handbook of the fundamentals of financial decision making:
  Part I}, pages 99--127. World Scientific, 2013.

\bibitem[Karni and Safra(1989)]{karni1989dynamic}
E.~Karni and Z.~Safra.
\newblock Dynamic consistency, revelations in auctions and the structure of
  preferences.
\newblock \emph{The Review of Economic Studies}, 56\penalty0 (3):\penalty0
  421--433, 1989.

\bibitem[Liu(2009)]{liu2009redundant}
Q.~Liu.
\newblock On redundant types and {B}ayesian formulation of incomplete
  information.
\newblock \emph{Journal of Economic Theory}, 144\penalty0 (5):\penalty0
  2115--2145, 2009.

\bibitem[Mas-Colell et~al.(1995)Mas-Colell, Whinston, Green,
  et~al.]{mas1995microeconomic}
A.~Mas-Colell, M.~D. Whinston, J.~R. Green, et~al.
\newblock \emph{Microeconomic theory}, volume~1.
\newblock Oxford university press New York, 1995.

\bibitem[Mookherjee and Tsumagari(2014)]{mookherjee2014mechanism}
D.~Mookherjee and M.~Tsumagari.
\newblock Mechanism design with communication constraints.
\newblock \emph{Journal of Political Economy}, 122\penalty0 (5):\penalty0
  1094--1129, 2014.

\bibitem[Myerson(1979)]{myerson1979incentive}
R.~B. Myerson.
\newblock Incentive compatibility and the bargaining problem.
\newblock \emph{Econometrica: journal of the Econometric Society}, pages
  61--73, 1979.

\bibitem[Myerson(1981)]{myerson1981optimal}
R.~B. Myerson.
\newblock Optimal auction design.
\newblock \emph{Mathematics of operations research}, 6\penalty0 (1):\penalty0
  58--73, 1981.

\bibitem[Myerson(1982)]{myerson1982optimal}
R.~B. Myerson.
\newblock Optimal coordination mechanisms in generalized principal--agent
  problems.
\newblock \emph{Journal of mathematical economics}, 10\penalty0 (1):\penalty0
  67--81, 1982.

\bibitem[Myerson(1986)]{myerson1986multistage}
R.~B. Myerson.
\newblock Multistage games with communication.
\newblock \emph{Econometrica: Journal of the Econometric Society}, pages
  323--358, 1986.

\bibitem[Myerson(2004)]{myerson2004comments}
R.~B. Myerson.
\newblock Comments on ``{G}ames with {I}ncomplete {I}nformation {P}layed by
  `{B}ayesian' {P}layers, {I}--{III} {H}arsanyi's {G}ames with {I}ncomplete
  {I}nformation''.
\newblock \emph{Management Science}, 50\penalty0 (12\_supplement):\penalty0
  1818--1824, 2004.

\bibitem[Pavlou and Dimoka(2006)]{pavlou2006nature}
P.~A. Pavlou and A.~Dimoka.
\newblock The nature and role of feedback text comments in online marketplaces:
  Implications for trust building, price premiums, and seller differentiation.
\newblock \emph{Information Systems Research}, 17\penalty0 (4):\penalty0
  392--414, 2006.

\bibitem[Phade and Anantharam(2018)]{phade2018learning}
S.~R. Phade and V.~Anantharam.
\newblock Learning in games with cumulative prospect theoretic preferences.
\newblock \emph{arXiv preprint arXiv:1804.08005}, 2018.

\bibitem[Prelec(1998)]{prelec1998probability}
D.~Prelec.
\newblock The probability weighting function.
\newblock \emph{Econometrica}, pages 497--527, 1998.

\bibitem[Schoemaker(1982)]{schoemaker1982expected}
P.~J. Schoemaker.
\newblock The expected utility model: Its variants, purposes, evidence and
  limitations.
\newblock \emph{Journal of economic literature}, pages 529--563, 1982.

\bibitem[Starmer(2000)]{starmer2000developments}
C.~Starmer.
\newblock Developments in non-expected utility theory: The hunt for a
  descriptive theory of choice under risk.
\newblock \emph{Journal of economic literature}, 38\penalty0 (2):\penalty0
  332--382, 2000.

\bibitem[Tversky and Kahneman(1992)]{tversky1992advances}
A.~Tversky and D.~Kahneman.
\newblock Advances in prospect theory: Cumulative representation of
  uncertainty.
\newblock \emph{Journal of Risk and uncertainty}, 5\penalty0 (4):\penalty0
  297--323, 1992.

\bibitem[Vohra(2011)]{vohra2011mechanism}
R.~V. Vohra.
\newblock \emph{Mechanism design: {A} linear programming approach}, volume~47.
\newblock Cambridge University Press, 2011.

\bibitem[Wakker(2010)]{wakker2010prospect}
P.~P. Wakker.
\newblock \emph{Prospect theory: For risk and ambiguity}.
\newblock Cambridge university press, 2010.

\bibitem[Williams(2008)]{williams2008communication}
S.~R. Williams.
\newblock \emph{Communication in mechanism design: A differential approach}.
\newblock Cambridge University Press, 2008.

\end{thebibliography}

\end{document}